\documentclass[11pt,draftcls,onecolumn]{IEEEtran}
\usepackage{graphicx}
\usepackage{amsmath}
\usepackage{mathrsfs}
\usepackage{mathtools}
\usepackage{amssymb}
\usepackage{float}
\usepackage{url}
\usepackage{verbatim}
\usepackage{dsfont}
\usepackage{enumerate}
\usepackage{layout}

\usepackage[left=20mm,right=20mm,top=24mm,bottom=24mm]{geometry}

\newtheorem{theorem}{Theorem}
\newtheorem{lemma}{Lemma}
\newtheorem{proposition}{Proposition}
\newtheorem{corollary}{Corollary}

\newtheorem{definition}{Definition}
\newtheorem{example}{Example}
\newtheorem{remark}{Remark}
\newcommand{\prob}{\ensuremath{\mathbb{P}}}

\newcommand{\Reals}{\ensuremath{\mathbb{R}}}
\newcommand{\expectation}{\ensuremath{\mathbb{E}}}
\newcommand{\set}{\ensuremath{\mathcal}}

\begin{document}
\title{\huge{On $f$-Divergences: Integral Representations, Local Behavior, and Inequalities}}
\author{Igal Sason
\thanks{
I. Sason is  with the Department of Electrical Engineering, Technion--Israel
Institute of Technology, Haifa 32000, Israel (e-mail: sason@ee.technion.ac.il).}}

\maketitle

\begin{abstract}
This paper is focused on $f$-divergences, consisting of three
main contributions. The first one introduces integral representations of
a general $f$-divergence by means of the relative information spectrum.
The second part provides a new approach for the derivation of $f$-divergence
inequalities, and it exemplifies their utility in the setup of Bayesian
binary hypothesis testing. The last part of this paper further studies
the local behavior of $f$-divergences.
\end{abstract}


\section{Introduction}
\label{section: introduction}
Probability theory, information theory, learning theory,
statistical signal processing and other related disciplines,
greatly benefit from non-negative measures of dissimilarity (a.k.a.
divergence measures) between pairs of probability measures defined
on the same measurable space (see, e.g., \cite{Basseville13}, \cite{LieseV_book87},
\cite{LieseV_IT2006}, \cite{ReidW11}, \cite{Tsybakov09}, \cite{Vapnik98}, \cite{Verdu_book}).
An axiomatic characterization of information measures, including divergence
measures, was provided by Csisz\'ar \cite{Csiszar08}.
Many useful divergence measures belong to the set of $f$-divergences,
independently introduced by Ali and Silvey
\cite{AliS}, Csisz\'{a}r (\cite{Csiszar63}--\cite{Csiszar67b}), and Morimoto
\cite{Morimoto63} in the early sixties. The family of $f$-divergences
generalizes the relative entropy (a.k.a. the Kullback-Leibler
divergence) while also satisfying the data processing inequality
among other pleasing properties (see, e.g., \cite{LieseV_IT2006}
and references therein).

Integral representations of $f$-divergences serve to study properties
of these information measures, and they are also used to establish relations
among these divergences. An integral representation of
$f$-divergences, expressed by means of the DeGroot statistical
information, was provided in \cite{LieseV_IT2006} with a simplified
proof in \cite{Liese2012}.
The importance of this integral representation stems from the
operational meaning of the DeGroot statistical information \cite{DeGroot62},
which is strongly linked to Bayesian binary hypothesis testing.
Some earlier specialized versions
of this integral representation were introduced in \cite{CZK98}, \cite{FeldmanO89},
\cite{Guttenbrunner92}, \cite{OV93} and \cite{Torgersen}, and a variation
of it also appears in \cite[Section~5.B]{ISSV16}.
Implications of the integral
representation of $f$-divergences, by means of the DeGroot statistical information,
include an alternative proof of the data processing inequality, and a study of
conditions for the sufficiency or $\varepsilon$-deficiency of observation channels
(\cite{LieseV_IT2006}, \cite{Liese2012}).

Since many distance measures of interest fall under the
paradigm of an $f$-divergence \cite{GibbsSu02}, bounds among $f$-divergences
are very useful in many instances such as the analysis of rates of
convergence and concentration of measure bounds,
hypothesis testing, testing goodness of fit, minimax risk in estimation and modeling,
strong data processing inequalities and contraction coefficients, etc.
Earlier studies developed systematic approaches to obtain $f$-divergence inequalities
while dealing with pairs of probability measures defined on arbitrary alphabets. A list
of some notable existing $f$-divergence inequalities is provided, e.g., in \cite[Section~3]{GibbsSu02}
and \cite[Section~1]{ISSV16}.
State-of-the-art techniques which serve to derive bounds among
$f$-divergences include:
\begin{enumerate}[1)]
\item Moment inequalities which rely on log-convexity arguments (\cite{AnwarHP09}, \cite[Section~5.D]{ISSV16},
\cite{Simic07}, \cite{Simic08}, \cite{Simic09}, \cite{Simic15});
\item Inequalities which rely on the characterization of the exact locus of the joint range of
$f$-divergences \cite{HarremoesV_2011};
\item $f$-divergence inequalities via functional domination
(\cite{SV16a}, \cite[Section~3]{ISSV16}, \cite{Taneja05b}, \cite{Taneja13});
\item Sharp $f$-divergence inequalities by using numerical tools for maximizing or minimizing an
$f$-divergence subject to a finite number of constraints on other $f$-divergences \cite{GSS_IT14};
\item Inequalities which rely
on powers of $f$-divergences defining a distance (\cite{EndresS03}, \cite{KafkaOV91}, \cite{LL2015}, \cite{Vajda09});
\item Vajda and Pinsker-type inequalities for $f$-divergences (\cite{Csiszar63}--\cite{Csiszar67b},
\cite{Gilardoni10}, \cite[Section~6-7]{ISSV16}, \cite{ReidW11}, \cite{Topsoe_IT00});
\item Bounds among $f$-divergences when the relative information is bounded (\cite{Dragomir00a}, \cite{Dragomir00b},
\cite{Dragomir00c}, \cite{DragomirG_01}, \cite{Dragomir06},
\cite{KumarC04}, \cite{SV15}, \cite[Sections~4-5]{ISSV16}, \cite{Taneja05}), and reverse Pinsker inequalities
(\cite{Binette18}, \cite{SV15}, \cite[Section~6]{ISSV16});
\item Inequalities which rely on the minimum of an $f$-divergence for a given total variation
distance and related bounds (\cite{Gilardoni06}, \cite{Gilardoni06-cor}, \cite{Gilardoni10}, \cite[p.~115]{GSS_IT14},
\cite{Gushchin16}, \cite{ReidW11}, \cite{IS15}, \cite{IS16}, \cite{Vajda09});
\item Bounds among $f$-divergences (or functions of $f$-divergences such as the R\'enyi divergence) via integral
representations of these divergence measures \cite[Section~8]{ISSV16};
\item Inequalities which rely on variational representations of $f$-divergences (e.g., \cite[Section~2]{Liu17}).
\end{enumerate}

Following earlier studies of the local behavior of $f$-divergences and their asymptotic properties (see
related results by Csisz\'{a}r and Shields \cite[Theorem~4.1]{CsiszarS_FnT}, Pardo and Vajda \cite[Section~3]{PardoV03},
and Sason and V\'{e}rdu \cite[Section~3.F]{ISSV16}), it is known that the local behavior of $f$-divergences
scales like the chi-square divergence (up to a scaling factor which depends on $f$) provided that the
first distribution approaches the reference measure in a certain strong sense. The study of the local
behavior of $f$-divergences is an important aspect of their properties, and we further study it in this work.

This paper considers properties of $f$-divergences, while first introducing in
Section~\ref{section: preliminaries} the basic definitions and notation needed,
and in particular the various measures of dissimilarity between probability measures
used throughout this paper. The presentation of our new results is then structured as follows:

Section~\ref{section: integral representation} is focused on
the derivation of new integral representations of $f$-divergences,
expressed as a function of the relative information spectrum of
the pair of probability measures, and the convex function $f$.
The novelty of Section~\ref{section: integral representation} is in
the unified approach which leads to integral representations of $f$-divergences
by means of the relative information spectrum, where the latter cumulative distribution function
plays an important role in information theory and statistical decision theory (see, e.g., \cite{Liu17} and
\cite{Verdu_book}). Particular integral representations of the type
of results introduced in Section~\ref{section: integral representation}
have been recently derived by Sason and Verd\'{u} in a case-by-case basis
for some $f$-divergences (see \cite[Theorems~13 and 32]{ISSV16}), while
lacking the approach which is developed in Section~\ref{section: integral representation}
for general $f$-divergences. In essence,
an $f$-divergence $D_f(P\|Q)$ is expressed in Section~\ref{section: integral representation}
as an inner product of a simple function of the relative
information spectrum (depending only on the probability measures $P$ and $Q$),
and a non-negative weight function $\omega_f \colon (0, \infty) \mapsto [0, \infty)$
which only depends on $f$. This kind of representation, followed by a generalized result,
serves to provide new integral representations of various useful $f$-divergences.
This also enables in Section~\ref{section: integral representation}
to characterize the interplay between the DeGroot statistical information (or
between another useful family of $f$-divergence, named the $E_\gamma$ divergence with $\gamma \geq 1$)
and the relative information spectrum.

Section~\ref{section: new inequalities} provides a new approach for the
derivation of $f$-divergence inequalities, where an arbitrary $f$-divergence
is lower bounded by means of the $E_\gamma$ divergence \cite{PPV10}
or the DeGroot statistical information \cite{DeGroot62}.
The approach used in Section~\ref{section: new inequalities} yields several
generalizations of the Bretagnole-Huber inequality \cite{BretagnolleH79},
which provides a closed-form and simple upper bound on the total variation
distance as a function of the relative entropy; the Bretagnole-Huber inequality
has been proved to be useful, e.g., in the context of lower bounding the minimax
risk in non-parametric estimation (see, e.g., \cite[pp.~89--90, 94]{Tsybakov09}),
and in the problem of density estimation (see, e.g., \cite[Section~1.6]{Vapnik98}).
Although Vajda's tight lower bound in \cite{Vajda70} is slightly tighter everywhere
than the Bretagnole-Huber inequality, our motivation for the generalization of
the latter bound is justified later in this paper.
The utility of the new inequalities is exemplified in the setup of Bayesian binary
hypothesis testing.

Section~\ref{section: local behavior} finally derives new results on the local behavior
of $f$-divergences, i.e., the characterization of their scaling when the pair of probability
measures are sufficiently close to each other. The starting point of our analysis in
Section~\ref{section: local behavior} relies on the analysis in \cite[Section~3]{PardoV03},
regarding the asymptotic properties of $f$-divergences.

The reading of Sections~\ref{section: integral representation}--\ref{section: local behavior}
can be done in any order since the analysis in these sections is independent.

\section{Preliminaries and Notation}
\label{section: preliminaries}

We assume throughout that the probability measures $P$ and $Q$ are
defined on a common measurable space $(\set{A}, \mathscr{F})$, and $P \ll Q$
denotes that $P$ is {\em absolutely continuous} with respect to $Q$, namely
there is no event $\set{F} \in \mathscr{F}$ such that $P(\set{F}) > 0 = Q(\set{F})$.

\begin{definition}
\label{def:RI}
The {\em relative information} provided by $a \in \set{A}$
according to $(P,Q)$, where $P \ll Q$, is given by
\begin{align}  \label{eq:RI}
\imath_{P\|Q}(a) := \log \, \frac{\text{d}P}{\text{d}Q} \, (a).
\end{align}
More generally, even if $P \not\ll Q$, let $R$ be an arbitrary dominating probability measure such that
$P,Q \ll R$ (e.g., $R = \tfrac12 (P+Q)$); irrespectively of the choice of $R$, the relative information
is defined to be
\begin{align}  \label{eq2:RI}
\imath_{P\|Q}(a) := \imath_{P\|R}(a) - \imath_{Q\|R}(a), \quad a \in \set{A}.
\end{align}
\end{definition}
The following asymmetry property follows from \eqref{eq2:RI}:
\begin{align}
\label{eq: asymmetry RI}
\imath_{P\|Q} = -\imath_{Q\|P}.
\end{align}

\begin{definition}  \label{def:RIS}
The {\em relative information spectrum} is the cumulative distribution function
\begin{align} \label{eq:RIS}
\mathds{F}_{P \| Q}(x) = \prob\bigl[\imath_{P\|Q}(X) \leq x \bigr], \quad x \in \Reals, \; X \sim P.
\end{align}
The {\em relative entropy} is the expected valued of the relative information when it is distributed according to
$P$:
\begin{align} \label{eq: KL div}
D(P\|Q) := \mathbb{E}\bigl[\imath_{P\|Q}(X)\bigr], \quad X \sim P.
\end{align}
\end{definition}

Throughout this paper, $\set{C}$ denotes the set of convex functions $f \colon (0, \infty) \mapsto \Reals$
with $f(1)=0$. Hence, the function $f \equiv 0$ is in $\set{C}$;
if $f \in \set{C}$, then $a f \in \set{C}$ for all $a>0$;
and if $f,g \in \set{C}$, then $f+g \in \set{C}$.
We next provide a general definition for the family of $f$-divergences (see \cite[p.~4398]{LieseV_IT2006}).
\begin{definition} \label{def:fD} ($f$-divergence \cite{AliS, Csiszar63, Csiszar67a})
Let $P$ and $Q$ be probability measures, let $\mu$ be a dominating measure of $P$ and $Q$
(i.e., $P, Q \ll \mu$; e.g., $\mu = P+Q$),
and let $p := \frac{\text{d}P}{\text{d}\mu}$ and $q := \frac{\text{d}Q}{\text{d}\mu}$.
The {\em $f$-divergence from $P$ to $Q$} is given, independently of $\mu$, by
\begin{align} \label{eq:fD}
D_f(P\|Q) := \int q \, f \Bigl(\frac{p}{q}\Bigr) \, \text{d}\mu,
\end{align}
where
\begin{align}
& f(0) := \underset{t \downarrow 0}{\lim} \, f(t), \\
& 0 f\Bigl(\frac{0}{0}\Bigr) := 0, \\
& 0 f\Bigl(\frac{a}{0}\Bigr) := \lim_{t \downarrow 0} \, t f\Bigl(\frac{a}{t}\Bigr) = a \lim_{u \to \infty} \frac{f(u)}{u}, \quad a>0.
\end{align}
\end{definition}

We rely in this paper on the following properties of $f$-divergences:
\begin{proposition} \label{proposition: uniqueness}
Let $f, g \in \set{C}$. The following conditions are equivalent:
\begin{enumerate}[1)]
\item
\begin{align}
D_f(P\|Q) = D_g(P\|Q), \quad \forall \, P, Q;
\end{align}
\item
there exists a constant $c \in \Reals$ such that
\begin{align} \label{eq: invariance of f-div}
f(t) - g(t) = c \; (t-1), \quad \forall \, t \in (0, \infty).
\end{align}
\end{enumerate}
\end{proposition}

\begin{proposition}  \label{proposition: conjugate}
Let $f \in \set{C}$, and let
$f^\ast \colon (0, \infty) \mapsto \Reals$ be the {\em conjugate function}, given by
\begin{align} \label{eq: conjugate f}
f^\ast(t) = t \, f\left(\tfrac1t\right)
\end{align}
for $t > 0$. Then,
\begin{enumerate}[1)]
\item $f^\ast \in \set{C}$;
\item $f^{\ast\ast} = f$;
\item for every pair of probability measures $(P,Q)$,
\begin{align} \label{eq: Df and Df^ast}
D_f(P \| Q) = D_{f^\ast}(Q \| P).
\end{align}
\end{enumerate}
\end{proposition}

By an analytic extension of $f^\ast$ in \eqref{eq: conjugate f} at $t=0$, let
\begin{align}
\label{eq: conjugate f at 0}
f^\ast(0) := \lim_{t \downarrow 0} f^\ast(t) = \lim_{u \to \infty} \frac{f(u)}{u}.
\end{align}
Note that the convexity of $f^\ast$ implies that $f^\ast(0) \in (-\infty, \infty]$.
In continuation to Definition~\ref{def:fD}, we get
\begin{align} \label{eq:fD2}
D_f(P\|Q) &= \int q \; f\left(\frac{p}{q}\right) \, \text{d}\mu \\
\label{eq2:fD2}
& = \int_{\{pq > 0\}} q \, f \left( \frac{p}{q} \right) \,\mathrm{d} \mu
+  Q( p = 0 ) \, f(0) + P (q = 0) \, f^\ast(0)
\end{align}
with the convention in \eqref{eq2:fD2} that $0 \cdot \infty = 0$.

\vspace*{0.2cm}
We refer in this paper to the following $f$-divergences:
\begin{enumerate}[1)]
\item {\em Relative entropy}:
\begin{align}
\label{eq: KL divergence}
D(P\|Q) &= D_f(P\|Q).
\end{align}
with
\begin{align}
\label{eq: f for KL}
f(t) = t\, \log t, \quad t>0.
\end{align}

\item {\em Jeffrey's divergence} \cite{jeffreys46}:
\begin{align}
\label{eq1: jeffreys}
J(P\|Q) &:= D( P \| Q) + D(Q\|P) \\
\label{eq2: jeffreys}
& \, = D_f(P\|Q)
\end{align}
with
\begin{align}
\label{f - jeffreys}
f(t) = (t-1) \, \log t,  \quad t>0.
\end{align}

\item {\em Hellinger divergence of order $\alpha \in (0,1) \cup (1, \infty)$}
\cite[Definition~2.10]{LieseV_book87}:
\begin{align} \label{eq: Hel-divergence}
\mathscr{H}_{\alpha}(P \| Q) = D_{f_\alpha}(P \| Q)
\end{align}
with
\begin{align} \label{eq: H as fD}
f_\alpha(t) = \frac{t^\alpha-1}{\alpha-1}, \quad t > 0.
\end{align}
Some of the significance of the Hellinger divergence stems from the following facts:
\begin{enumerate}[a)]
\item
The analytic extension of $\mathscr{H}_{\alpha}(P \| Q)$ at $\alpha=1$ yields
\begin{align}
\label{eq1: KL}
D(P \| Q) = H_1(P\| Q) \, \log e.
\end{align}
\item
The {\em chi-squared divergence} \cite{Pearson1900x} is
the second order Hellinger divergence (see, e.g.,
\cite[p.~48]{LeCam86}), i.e.,
\begin{align}
\label{eq2: chi^2}
\chi^2(P \| Q) = \mathscr{H}_2(P \| Q).
\end{align}
Note that, due to Proposition~\ref{proposition: uniqueness},
\begin{align} \label{eq: chi squared}
\chi^2(P \| Q) = D_f(P\|Q),
\end{align}
where $f \colon (0, \infty) \mapsto \Reals$ can be defined as
\begin{align} \label{eq: f for chi^2}
f(t) = (t-1)^2, \quad t>0.
\end{align}
\item The \textit{squared Hellinger distance} (see, e.g., \cite[p.~47]{LeCam86}),
denoted by $\mathscr{H}^2(P \| Q)$, satisfies the identity
\begin{align}
\label{eq3: Sq Hel}
\mathscr{H}^2(P \| Q) = \tfrac12 \, \mathscr{H}_{\frac12}(P \| Q).
\end{align}
\item
The {\em Bhattacharyya distance} \cite{Kailath67}, denoted by
$B(P\|Q)$, satisfies
\begin{align}
\label{eq: B dist}
B(P\|Q) = \log \frac1{1 - \mathscr{H}^2(P \| Q)}.
\end{align}
\item
The {\em R\'enyi divergence} of order $\alpha \in (0,1) \cup (1,\infty)$ is a
one-to-one transformation of the Hellinger divergence of the same order \cite[(14)]{Csiszar66}:
\begin{align}\label{renyimeetshellinger}
D_\alpha(P \| Q ) = \frac1{\alpha -1} \; \log \bigl( 1 + (\alpha - 1)
\, \mathscr{H}_\alpha(P \| Q) \bigr).
\end{align}
\item
The {\em Alpha-divergence} of order $\alpha$, as it is defined in \cite{AmariN00} and \cite[(4)]{CichockiA10}, is
a generalized relative entropy which (up to a scaling factor) is equal to the Hellinger divergence of the
same order $\alpha$. More explicitly,
\begin{align}
D_{\text{A}}^{(\alpha)}(P\|Q) = \frac1\alpha \, \mathscr{H}_{\alpha}(P \| Q),
\end{align}
where $D_{\text{A}}^{(\alpha)}(\cdot \| \cdot)$ denotes the Alpha-divergence of order $\alpha$.
Note, however, that the Beta and Gamma-divergences in \cite{CichockiA10}, as well as the
generalized divergences in \cite{CichockiCA11} and \cite{CichockiCA15}, are
not $f$-divergences in general.
\end{enumerate}

\item {\em $\chi^s$ divergence for $s \geq 1$} \cite[(2.31)]{LieseV_book87}, and the {\em total variation distance}:
The function
\begin{align}
\label{eq: f - chi^s div}
f_s(t) = |t-1|^s, \quad t>0
\end{align}
results in
\begin{align}
\label{eq: chi^s div}
\chi^s(P\|Q) &=  D_{f_s}(P\|Q).
\end{align}
Specifically, for $s=1$, let
\begin{align} \label{eq: f-TV}
f(t):=f_1(t)=|t-1|, \quad t>0,
\end{align}
and the total variation distance is expressed as an $f$-divergence:
\begin{align}
\label{eq1: TV distance}
|P-Q| &=  D_f(P\|Q).
\end{align}

\item {\em Triangular Discrimination \cite{Topsoe_IT00} (a.k.a. Vincze-Le Cam distance)}:
\begin{align} \label{eq:delta}
\Delta(P\|Q) = D_f(P\|Q)
\end{align}
with
\begin{align} \label{eq:tridiv}
f(t) = \frac{(t-1)^2}{t+1}, \quad t > 0.
\end{align}
Note that
\begin{align}
\label{eq: Delta and chi^2}
\tfrac12 \, \Delta  (P \| Q) & = \chi^2 (P \, \| \, \tfrac12 P + \tfrac12 Q)
= \chi^2 (Q \, \| \, \tfrac12 P + \tfrac12 Q).
\end{align}

\item {\em Lin's measure} \cite[(4.1)]{Lin91}:
\begin{align} \label{eq: Lin91}
L_\theta(P\|Q) &:= H\bigl(\theta P + (1-\theta)Q\bigr) - \theta H(P) - (1-\theta) H(Q) \\
\label{eq2: Lin91}
& \, = \theta \, D\bigl(P \, \| \, \theta P + (1-\theta)Q\bigr) + (1-\theta) \, D\bigl(Q \, \| \, \theta P + (1-\theta)Q \bigr),
\end{align}
for $\theta \in [0,1]$. This measure can be expressed by the following $f$-divergence:
\begin{align}
\label{eq: Lin div as Df}
L_\theta(P\|Q) = D_{f_\theta}(P\|Q),
\end{align}
with
\begin{align}
\label{eq: f of Lin div}
f_{\theta}(t) := \theta t \log t - \bigl(\theta t + 1-\theta \bigr) \, \log \bigl(\theta t + 1-\theta \bigr), \quad t>0.
\end{align}
The special case of \eqref{eq: Lin div as Df} with $\theta = \tfrac12$ gives the
{\em Jensen-Shannon divergence} (a.k.a. capacitory discrimination):
\begin{align}
\label{eq:js1}
\mathrm{JS}(P\|Q) & := L_{\frac12}(P\|Q) \\
\label{eq:js2}
&= \tfrac12 D\bigl(P \, \| \, \tfrac12 P + \tfrac12 Q\bigr) + \tfrac12 D\bigl(Q \, \| \, \tfrac12 P + \tfrac12 Q \bigr).
\end{align}

\item {\em $E_\gamma$ divergence} \cite[p.~2314]{PPV10}: For $\gamma \geq 1$,
\begin{align}
\label{eq1:E_gamma}
E_{\gamma}(P\|Q) &:= \max_{\set{U} \in \mathscr{F}} \bigl( P(\set{U}) - \gamma \, Q(\set{U}) \bigr) \\
\label{eq2:E_gamma}
& \; = \prob[\imath_{P\|Q}(X) > \log \gamma] - \gamma \, \prob[\imath_{P\|Q}(Y) > \log \gamma]
\end{align}
with $X \sim P$ and $Y \sim Q$, and where \eqref{eq2:E_gamma} follows from the Neyman-Pearson
lemma. The $E_\gamma$ divergence can be identified as an $f$-divergence:
\begin{align} \label{eq:Eg f-div}
E_{\gamma}(P \| Q) = D_{f_\gamma}(P\|Q)
\end{align}
with
\begin{align} \label{eq: f for EG}
f_\gamma(t) := (t-\gamma)^+, \quad t > 0
\end{align}
where $(x)^+ := \max\{x,0\}$.
The following relation to the total variation distance holds:
\begin{align} \label{eq:EG-TV}
E_1(P\|Q) = \tfrac12 \, |P-Q|.
\end{align}

\item {\em DeGroot statistical information} (\cite{DeGroot62}, \cite{LieseV_IT2006}):
For $\omega \in (0,1)$,
\begin{align} \label{eq:DG f-div}
\mathcal{I}_\omega(P\|Q) = D_{\phi_\omega}(P\|Q)
\end{align}
with
\begin{align} \label{eq: f for DG}
\phi_\omega(t) = \min \{\omega, 1-\omega\}
- \min \{\omega t, 1-\omega\}, \quad t > 0.
\end{align}
The following relation to the total variation distance holds:
\begin{align} \label{eq:DG-TV}
\mathcal{I}_{\frac12}(P\|Q) = \tfrac14 \, |P-Q|,
\end{align}
and the DeGroot statistical information and the $E_\gamma$ divergence
are related as follows \cite[(384)]{ISSV16}:
\begin{align}
\label{eq: DG-EG}
\mathcal{I}_\omega(P\|Q) =
\begin{dcases}
\omega \, E_{\frac{1-\omega}{\omega}}(P\|Q), & \quad \mbox{$\omega \in \bigl(0, \tfrac12\bigr]$,} \\[0.2cm]
(1-\omega) \, E_{\frac{\omega}{1-\omega}}(Q\|P), & \quad \mbox{$\omega \in \bigl(\tfrac12, 1)$.}
\end{dcases}
\end{align}
\end{enumerate}

\vspace*{0.2cm}
\section{New Integral Representations of $f$-divergences}
\label{section: integral representation}

The main result in this section provides new integral representations
of $f$-divergences as a function of the relative information spectrum
(see Definition~\ref{def:RIS}). The reader is referred to other integral
representations (see \cite[Section~2]{Liese2012}, \cite[Section~5]{ReidW11},
\cite[Section~5.B]{ISSV16}, and references therein), expressing a general
$f$-divergence by means of the DeGroot statistical information or the
$E_\gamma$ divergence.

\begin{lemma} \label{lemma: g}
Let $f \in \set{C}$ be a strictly convex function at~1. Let $g \colon \Reals \mapsto \Reals$ be defined as
\begin{align} \label{eq:g}
g(x) := \exp(-x) \, f\bigl(\exp(x)\bigr) - f'_{+}(1) \, \bigl(1 - \exp(-x) \bigr), \qquad x \in \Reals
\end{align}
where $f'_{+}(1)$ denotes the right-hand derivative of $f$ at 1 (due to the convexity of $f$ on $(0, \infty)$,
it exists and it is finite). Then, the function $g$ is non-negative, it is strictly monotonically decreasing
on $(-\infty, 0]$, and it is strictly monotonically increasing on $[0, \infty)$ with $g(0)=0$.
\end{lemma}

\begin{proof}
For any function $u \in \set{C}$, let $\widetilde{u} \in \set{C}$ be given by
\begin{align} \label{eq:widetilde u}
\widetilde{u}(t) = u(t) - u'_{+}(1) (t-1), \quad t \in (0, \infty),
\end{align}
and let $u^\ast \in \set{C}$ be the conjugate function, as given in \eqref{eq: conjugate f}.
The function $g$ in \eqref{eq:g} can be expressed in the form
\begin{align}
\label{eq:decompose g}
g(x) = (\widetilde{f})^{\ast} \bigl( \exp(-x) \bigr), \quad x \in \Reals,
\end{align}
as it is next verified. For $t>0$, we get from \eqref{eq: conjugate f} and \eqref{eq:widetilde u},
\begin{align}
\label{eq:verify g}
(\widetilde{f})^{\ast}(t) = t \widetilde{f}\left(\frac1t\right)
= t f\left( \frac1t \right) + f'_{+}(1) \, (t-1),
\end{align}
and the substitution $t := \exp(-x)$ for $x \in \Reals$
yields \eqref{eq:decompose g} in view of \eqref{eq:g}.

By assumption, $f \in \set{C}$ is strictly convex at~1, and therefore these properties are inherited to
$\widetilde{f}$. Since also $\widetilde{f}(1) = \widetilde{f}'(1) = 0$, it follows from \cite[Theorem~3]{LieseV_IT2006}
that both $\widetilde{f}$ and $\widetilde{f}^\ast$ are non-negative on $(0, \infty)$, and they are also
strictly monotonically decreasing on $(0,1]$. Hence, from \eqref{eq: conjugate f}, it follows that the
function $(\widetilde{f})^\ast$ is strictly monotonically increasing on $[1, \infty)$.
Finally, the claimed properties of the function $g$ follow from \eqref{eq:decompose g}, and in view
of the fact that the function $(\widetilde{f})^\ast$ is non-negative with $(\widetilde{f})^\ast(1)=0$,
strictly monotonically decreasing on $(0,1]$ and strictly monotonically increasing on $[1, \infty)$.
\end{proof}

\begin{lemma} \label{lemma: 1st int. for f-div}
Let $f \in \set{C}$ be a strictly convex function at~1, and let $g \colon \Reals \mapsto \Reals$ be as
in \eqref{eq:g}. Let
\begin{align}
\label{eq:a}
a & := \lim_{x \to \infty} g(x) \in (0, \infty], \\
\label{eq:b}
b & := \lim_{x \to -\infty} g(x) \in (0, \infty],
\end{align}
and let $\ell_1 \colon [0, a) \mapsto [0, \infty)$ and $\ell_2 \colon [0, b) \mapsto (-\infty, 0]$ be
the two inverse functions of $g$. Then,
\begin{align}  \label{eq: 1st int. for f-div}
D_f(P \| Q) = \int_0^a \bigl[ 1 - \mathds{F}_{P \| Q}\bigl( \ell_1(t) \bigr) \bigr] \, \mathrm{d}t
+ \int_0^b \mathds{F}_{P \| Q}\bigl( \ell_2(t) \bigr) \, \mathrm{d}t.
\end{align}
\end{lemma}

\begin{proof}
In view of Lemma~\ref{lemma: g}, it follows that $\ell_1 \colon [0, a) \mapsto [0, \infty)$ is strictly monotonically increasing
and $\ell_2 \colon [0, b) \mapsto (-\infty, 0]$ is strictly monotonically decreasing with $\ell_1(0) = \ell_2(0) = 0$.

Let $X \sim P$, and let $V := \exp \bigl( \imath_{P\|Q}(X) \bigr)$. Then, we have
\begin{align}
\label{eq: linear shift}
D_f(P\|Q) &= D_{\widetilde{f}}(P \| Q) \\
\label{eq: P and Q switched}
&= D_{(\widetilde{f})^\ast}(Q\|P) \\
&= \int (\widetilde{f})^\ast \bigl( \exp\bigl(\imath_{Q\|P}(x) \bigr) \bigr) \, \mathrm{d}P(x) \\
\label{eq: plus/minus RI}
&= \int (\widetilde{f})^\ast \bigl( \exp\bigl(-\imath_{P\|Q}(x) \bigr) \bigr) \, \mathrm{d}P(x) \\
\label{eq: by the expression for g}
&= \int g \bigl( \imath_{P\|Q}(x) \bigr) \, \mathrm{d}P(x) \\
\label{eq: by V}
&= \expectation\bigl[ g(V) \bigr] \\
\label{eq: expectation of non-negative RV}
&= \int_0^{\infty} \prob \bigl[ g(V) > t \bigr] \, \mathrm{d}t \\
\label{eq: sum pf prob}
&= \int_0^a \prob \bigl[V \geq 0, \, g(V) > t \bigr] \, \mathrm{d}t
   + \int_0^b \prob \bigl[ V < 0, \, g(V)>t \bigr] \, \mathrm{d}t \\
\label{eq: ell's}
&= \int_0^a \prob \bigl[V > \ell_1(t) \bigr] \, \mathrm{d}t
   + \int_0^b \prob \bigl[ V \leq \ell_2(t) \bigr] \, \mathrm{d}t \\
\label{eq: expression with RIS}
&= \int_0^a \bigl[ 1 - \mathds{F}_{P \| Q}\bigl( \ell_1(t) \bigr) \bigr] \, \mathrm{d}t
   + \int_0^b \mathds{F}_{P \| Q}\bigl( \ell_2(t) \bigr) \, \mathrm{d}t
\end{align}
where \eqref{eq: linear shift} relies on Proposition~\ref{proposition: uniqueness};
\eqref{eq: P and Q switched} relies on Proposition~\ref{proposition: conjugate};
\eqref{eq: plus/minus RI} follows from \eqref{eq: asymmetry RI};
\eqref{eq: by the expression for g} follows from \eqref{eq:decompose g};
\eqref{eq: by V} holds by the definition of the random variable $V$;
\eqref{eq: expectation of non-negative RV} holds since, in view of Lemma~\ref{lemma: g},
$Z := g(V) \geq 0$, and $\expectation[Z] = \int_0^{\infty} \prob[Z > t] \, \mathrm{d}t$
for any non-negative random variable $Z$;
\eqref{eq: sum pf prob} holds in view of the monotonicity properties of $g$ in Lemma~\ref{lemma: g},
the definition of $a$ and $b$ in \eqref{eq:a} and \eqref{eq:b}, and by expressing the event
$\{g(V) > t\}$ as a union of two disjoint events;
\eqref{eq: ell's} holds again by the monotonicity properties of $g$ in Lemma~\ref{lemma: g},
and by the definition of its two inverse functions $\ell_1$ and $\ell_2$ as above;
in \eqref{eq: expectation of non-negative RV}--\eqref{eq: ell's}
we are free to substitute $>$ by $\geq$, and $<$ by $\leq$; finally, \eqref{eq: expression with RIS}
holds by the definition of the relative information spectrum in \eqref{eq:RIS}.
\end{proof}

\begin{remark} \label{remark: invariance of g}
The function $g \colon \Reals \mapsto \Reals$ in \eqref{eq:g} is invariant to the mapping
$f(t) \mapsto f(t) + c \, (t-1)$, for $t>0$, with an arbitrary $c \in \Reals$. This invariance
of $g$ (and, hence, also the invariance of its inverse functions $\ell_1$ and $\ell_2$) is
well expected in view of Proposition~\ref{proposition: uniqueness} and Lemma~\ref{lemma: 1st int. for f-div}.
\end{remark}

\begin{example} \label{example: chi-squared divergence}
For the chi-squared divergence in \eqref{eq: chi squared}, letting $f$ be as in \eqref{eq: f for chi^2},
it follows from \eqref{eq:g} that
\begin{align}
g(x) = 4 \sinh^2\left(\tfrac1{2 \log e} \, x \right), \quad x \in \Reals,
\end{align}
which yields, from \eqref{eq:a} and \eqref{eq:b}, $a=b=\infty$. Calculation of the two inverse functions of $g$,
as defined in Lemma~\ref{lemma: 1st int. for f-div}, yields the following closed-form expression:
\begin{align} \label{ell_1,2 chi^2}
\ell_{1,2}(u) &= \pm 2 \log \left(\frac{u+\sqrt{u+4}}{2}\right), \quad u \geq 0.
\end{align}
Substituting \eqref{ell_1,2 chi^2} into \eqref{eq: 1st int. for f-div} provides an integral
representation of $\chi^2(P\|Q)$.
\end{example}

\begin{lemma} \label{lemma: identity with RIS}
\begin{align} \label{eq: int. identity with RIS}
\int_0^\infty \frac{\mathds{F}_{P \| Q}(\log \beta)}{\beta^2} \, \mathrm{d}\beta = 1.
\end{align}
\end{lemma}
\begin{proof}
Let $X \sim P$. Then, we have
\begin{align}
\int_0^\infty \frac{\mathds{F}_{P \| Q}(\log \beta)}{\beta^2} \, \mathrm{d}\beta
\label{eq1: int. identity with RIS}
& = \int_0^\infty \frac1{\beta^2} \, \prob[ \imath_{P\|Q}(X) \leq \log \beta ] \, \mathrm{d}\beta \\[0.1cm]
\label{eq2: int. identity with RIS}
& = \int_0^\infty \frac1{\beta^2} \, \prob\biggl[ \exp\bigl(\imath_{Q\|P}(X)\bigr) \geq \frac1\beta \biggr] \, \mathrm{d}\beta \\[0.1cm]
\label{eq3: int. identity with RIS}
& = \int_0^\infty \, \prob\bigl[ \exp\bigl(\imath_{Q\|P}(X)\bigr) \geq u \bigr] \, \mathrm{d}u \\[0.1cm]
\label{eq4: int. identity with RIS}
& = \expectation\bigl[ \exp\bigl(\imath_{Q\|P}(X)\bigr)\bigr] \\
\label{eq5: int. identity with RIS}
& = 1,
\end{align}
where \eqref{eq1: int. identity with RIS} holds by \eqref{eq:RIS};
\eqref{eq2: int. identity with RIS} follows from \eqref{eq: asymmetry RI};
\eqref{eq3: int. identity with RIS} holds by the substitution $u := \frac1\beta$;
\eqref{eq4: int. identity with RIS} holds since $\exp\bigl(\imath_{Q\|P}(X)\bigr) \geq 0$,
and finally \eqref{eq5: int. identity with RIS} holds since $X \sim P$.
\end{proof}

\begin{remark}
Unlike Example~\ref{example: chi-squared divergence}, in general, the inverse functions
$\ell_1$ and $\ell_2$ in Lemma~\ref{lemma: 1st int. for f-div} are not expressible in
closed form, motivating our next integral representation in Theorem~\ref{theorem: Int. rep.}.
\end{remark}

\vspace*{0.2cm}
The following theorem provides our main result in this section.
\begin{theorem} \label{theorem: Int. rep.}
The following integral representations of $f$-divergences hold:
\begin{enumerate}[a)]
\item
\label{theorem: int. rep. - part 1}
Let
\begin{itemize}
\item $f \in \set{C}$ be differentiable on $(0, \infty)$;
\item $w_f \colon (0, \infty) \mapsto [0, \infty)$ be the non-negative weight
function given, for $\beta>0$, by
\begin{align} \label{eq: weight function}
w_f(\beta) & := \frac1{\beta} \left| f'(\beta) - \frac{f(\beta) + f'(1)}{\beta} \right|;
\end{align}
\item the function $G_{P\|Q} \colon (0, \infty) \mapsto [0,1]$ be given by
\begin{align} \label{eq: G function}
G_{P\|Q}(\beta) :=
\begin{dcases}
1 - \mathds{F}_{P \| Q}(\log \beta), & \beta \in [1, \infty), \\[0.1cm]
\mathds{F}_{P \| Q}(\log \beta), & \beta \in (0,1).
\end{dcases}
\end{align}
\end{itemize}
Then,
\begin{align} \label{eq: new int rep Df}
D_f(P \| Q) = \langle w_f, \, G_{P\|Q} \rangle
= \int_0^{\infty} w_f(\beta) \, G_{P\|Q}(\beta) \, \mathrm{d}\beta.
\end{align}
\item  \label{theorem: int. rep. - part 2}
More generally, for an arbitrary $c \in \Reals$, let $\widetilde{w}_{f,c} \colon (0, \infty) \mapsto \Reals$
be a modified real-valued function defined as
\begin{align}  \label{eq: generalized w_f}
\widetilde{w}_{f,c}(\beta) := w_f(\beta) + \frac{c}{\beta^2} \, \bigl( 1\{\beta \geq 1\} - 1\{0 < \beta < 1\} \bigr).
\end{align}
Then,
\begin{align} \label{eq2: new int rep Df}
D_f(P \| Q) = \langle \widetilde{w}_{f,c}, \, G_{P\|Q} \rangle.
\end{align}
\end{enumerate}
\end{theorem}

\begin{proof}
We start by proving the special integral representation in \eqref{eq: new int rep Df},
and then extend our proof to the general representation in \eqref{eq2: new int rep Df}.
\begin{enumerate}[a)]
\item We first assume an additional requirement that $f$ is strictly
convex at~1. In view of Lemma~\ref{lemma: 1st int. for f-div},
\begin{align}
\label{eq: ell_1 of g}
& \ell_1 \bigl( g(u) \bigr) = u, \quad u \in [0, \infty), \\
\label{eq: ell_2 of g}
& \ell_2 \bigl( g(u) \bigr) = u, \quad u \in (-\infty, 0].
\end{align}
Since by assumption $f \in \set{C}$ is differentiable on $(0, \infty)$
and strictly convex at~1,
the function $g$ in \eqref{eq:g} is differentiable on $\Reals$.
In view of \eqref{eq: ell_1 of g} and \eqref{eq: ell_2 of g},
substituting $t := g \bigl( \log \beta \bigr)$ in \eqref{eq: 1st int. for f-div}
for $\beta > 0$ implies that
\begin{align} \label{eq: int with mod w_f}
D_f(P \| Q) = \int_1^\infty \bigl[ 1 - \mathds{F}_{P \| Q}\bigl( \log \beta \bigr) \bigr] \,
\overline{w}_f(\beta) \, \mathrm{d}\beta
- \int_0^1 \mathds{F}_{P \| Q}\bigl( \log \beta \bigr) \, \overline{w}_f(\beta) \, \mathrm{d}\beta,
\end{align}
where $\overline{w}_f \colon (0, \infty) \mapsto \Reals$ is given by
\begin{align} \label{eq: mod w_f - g}
\overline{w}_f(\beta) & := \frac{g'\bigl( \log \beta \bigr)}{\beta} \, \log \mathrm{e} \\
\label{eq: mod w_f - f}
& = \frac1\beta \left[ f'(\beta) - \frac{f(\beta) + f'(1)}{\beta} \right]
\end{align}
for $\beta>0$, where \eqref{eq: mod w_f - f} follows from \eqref{eq:g}.
Due to the monotonicity properties of $g$ in Lemma~\ref{lemma: g},
\eqref{eq: mod w_f - g} implies that $\overline{w}_f(\beta) \geq 0$
for $\beta \geq 1$, and $\overline{w}_f(\beta) < 0$ for
$\beta \in (0,1)$. Hence, the weight function $w_f$ in
\eqref{eq: weight function} satisfies
\begin{align} \label{eq: w_f/ mod w_f}
w_f(\beta) = \bigl| \overline{w}_f(\beta) \bigr| = \overline{w}_f(\beta) \, \bigl( 1\{\beta \geq 1\}
- 1\{0 < \beta < 1\} \bigr), \quad \beta>0.
\end{align}
The combination of \eqref{eq: G function}, \eqref{eq: int with mod w_f} and \eqref{eq: w_f/ mod w_f}
gives the required result in \eqref{eq: new int rep Df}.

\par
We now extend the result in \eqref{eq: new int rep Df} when $f \in \set{C}$ is differentiable on $(0, \infty)$,
but not necessarily strictly convex at~1.
To that end, let $s \colon (0, \infty) \mapsto \Reals$ be defined as
\begin{align} \label{eq: s function}
s(t) := f(t) + (t^2-1), \quad t>0.
\end{align}
This implies that $s \in \set{C}$ is differentiable on $(0, \infty)$, and it is also strictly convex at~1.
In view of the proof of \eqref{eq: new int rep Df} when $f$ is strict convexity of $f$ at~1, the application
of this result to the function $s$ in \eqref{eq: s function} yields
\begin{align}
\label{eq: apply for s}
D_s(P\|Q) = \langle w_s,  \, G_{P\|Q} \rangle.
\end{align}
In view of \eqref{eq:fD}, \eqref{eq: Hel-divergence}, \eqref{eq: H as fD}, \eqref{eq2: chi^2}
and \eqref{eq: s function},
\begin{align}
\label{eq: D_s}
D_s(P\|Q) = D_f(P\|Q) + \chi^2(P\|Q);
\end{align}
from \eqref{eq: weight function}, \eqref{eq: w_f/ mod w_f}, \eqref{eq: s function}
and the convexity and differentiability of $f \in \set{C}$, it follows that the weight function
$w_s \in (0, \infty) \mapsto [0, \infty)$ satisfies
\begin{align}
\label{eq: w_s}
w_s(\beta) =  w_f(\beta) + \left(1-\frac1{\beta^2}\right) \left(1\{\beta \geq 1\} - 1\{0<\beta<1\}\right)
\end{align}
for $\beta>0$. Furthermore, by applying the result in \eqref{eq: new int rep Df} to the chi-squared
divergence $\chi^2(P\|Q)$ in \eqref{eq2: chi^2} whose corresponding function $f_2(t) := t^2-1$ for $t>0$
is strictly convex at~1, we obtain
\begin{align}
\label{eq: application to chi^2}
\chi^2(P\|Q) = \int_0^\infty \left(1-\frac1{\beta^2}\right) \left(1\{\beta \geq 1\} - 1\{0<\beta<1\}\right)
\, G_{P\|Q}(\beta) \, \mathrm{d}\beta.
\end{align}
Finally, the combination of \eqref{eq: apply for s}--\eqref{eq: application to chi^2}, yields
$D_f(P\|Q) = \langle w_f,  \, G_{P\|Q} \rangle$; this asserts that \eqref{eq: new int rep Df} also holds
by relaxing the condition that $f$ is strictly convex at~1.

\item In view of \eqref{eq: G function}, \eqref{eq: new int rep Df} and \eqref{eq: generalized w_f},
in order to prove \eqref{eq2: new int rep Df} for an arbitrary $c \in \Reals$, it is required to prove
the identity
\begin{align} \label{eq7: identity RIS}
\int_1^{\infty} \frac{1-\mathds{F}_{P \| Q}\bigl( \log \beta \bigr)}{\beta^2} \, \mathrm{d}\beta =
\int_0^1 \frac{\mathds{F}_{P \| Q}\bigl( \log \beta \bigr)}{\beta^2} \, \mathrm{d}\beta.
\end{align}
Equality~\eqref{eq7: identity RIS} can be verified by Lemma~\ref{lemma: identity with RIS}:
by rearranging terms in \eqref{eq7: identity RIS}, we get the identity in \eqref{eq: int. identity with RIS}
(since $\int_1^{\infty} \frac{\mathrm{d}\beta}{\beta^2} = 1$).
\end{enumerate}
\end{proof}

\begin{remark}  \label{remark 2: w_f}
Due to the convexity of $f$, the absolute value in the right side of \eqref{eq: weight function}
is only needed for $\beta \in (0,1)$ (see \eqref{eq: mod w_f - f} and \eqref{eq: w_f/ mod w_f}).
Also, $w_f(1)=0$ since $f(1)=0$.
\end{remark}

\begin{remark} \label{remark 1: w_f}
The weight function $w_f$ only depends on $f$, and the function $G_{P\|Q}$ only depends on the
pair of probability measures $P$ and $Q$. In view of Proposition~\ref{proposition: uniqueness},
it follows that, for $f, g \in \set{C}$, the equality $w_f = w_g$ holds on $(0, \infty)$ if and
only if \eqref{eq: invariance of f-div} is satisfied with an arbitrary constant $c \in \Reals$.
It is indeed easy to verify that \eqref{eq: invariance of f-div} yields $w_f = w_g$ on $(0, \infty)$.
\end{remark}

\begin{remark} \label{remark: G function}
An equivalent way to write $G_{P\|Q}$ in \eqref{eq: G function} is
\begin{align} \label{eq2: G function}
G_{P\|Q}(\beta) =
\begin{dcases}
\prob\left[ \frac{\text{d}P}{\text{d}Q} \, (X) > \beta \right], & \beta \in [1, \infty) \\[0.2cm]
\prob\left[ \frac{\text{d}P}{\text{d}Q} \, (X) \leq \beta \right], & \beta \in (0, 1)
\end{dcases}
\end{align}
where $X \sim P$.
Hence, the function $G_{P\|Q} \colon (0, \infty) \mapsto [0,1]$ is monotonically increasing in $(0,1)$,
and it is monotonically decreasing in $[1, \infty)$; note that this function is in general discontinuous
at~1 unless $\mathds{F}_{P \| Q}(0) = \tfrac12$.
If $P \ll \gg Q$, then
\begin{align}
\lim_{\beta \downarrow 0} G_{P\|Q}(\beta) = \lim_{\beta \to \infty} G_{P\|Q}(\beta) = 0.
\end{align}
Note that if $P = Q$, then $G_{P\|Q}$ is zero everywhere, which is consistent with the fact that $D_f(P \| Q)=0$.
\end{remark}

\begin{remark} \label{remark: strict convexity}
In the proof of Theorem~\ref{theorem: Int. rep.}-\ref{theorem: int. rep. - part 1}),
the relaxation of the condition of strict convexity at~1 for a differentiable function $f \in \set{C}$ is crucial, e.g., for
the $\chi^s$ divergence with $s>2$. To clarify this claim, note that in view of \eqref{eq: f - chi^s div}, the function
$f_s \colon (0, \infty) \mapsto \Reals$ is differentiable if $s>1$, and $f_s \in \set{C}$ with $f_s'(1)=0$;
however, $f_s''(1)=0$ if $s>2$, so $f_s$ in not strictly convex at~1 unless $s \in [1,2]$.
\end{remark}

\begin{remark} \label{remark: utility of part 2}
Theorem~\ref{theorem: Int. rep.}-\ref{theorem: int. rep. - part 2}) with $c \neq 0$ enables, in some cases,
to simplify integral representations of $f$-divergences. This is next exemplified in the proof of
Theorem~\ref{theorem: some int. representations}.
\end{remark}

\par
Theorem~\ref{theorem: Int. rep.} yields integral representations for various $f$-divergences
and related measures; some of these representations were previously derived by Sason and Verd\'{u} in
\cite{ISSV16} in a case by case basis, without the unified approach of Theorem~\ref{theorem: Int. rep.}.
We next provide such integral representations.
Note that, for some $f$-divergences, the function $f \in \set{C}$ is not differentiable
on $(0, \infty)$; hence, Theorem~\ref{theorem: Int. rep.} is not necessarily directly applicable.

\begin{theorem} \label{theorem: some int. representations}
The following integral representations hold as a function of the relative information spectrum:
\begin{enumerate}[1)]
\item Relative entropy \cite[(219)]{ISSV16}:
\begin{align} \label{eq: int. rep. KL}
\tfrac1{\log e} \, D(P\|Q) &= \int_1^{\infty} \frac{1-\mathds{F}_{P\|Q}(\log \beta)}{\beta} \, \mathrm{d}\beta
- \int_0^1 \frac{\mathds{F}_{P\|Q}(\log \beta)}{\beta} \, \mathrm{d}\beta.
\end{align}

\item Hellinger divergence of order $\alpha \in (0,1) \cup (1, \infty)$ \cite[(434) and (437)]{ISSV16}:
\begin{align}
\label{eq: int. rep. Hel}
\mathscr{H}_{\alpha}(P \| Q) &=
\begin{dcases}
\frac1{1-\alpha} - \int_0^\infty \beta^{\alpha-2} \, \mathds{F}_{P\|Q}(\log \beta) \, \mathrm{d}\beta, &\; \alpha \in (0,1) \\[0.2cm]
\int_0^{\infty} \beta^{\alpha-2} \left(1 - \mathds{F}_{P\|Q}(\log \beta) \right)  \, \mathrm{d}\beta - \frac1{\alpha-1}, &\; \alpha \in (1, \infty).
\end{dcases}
\end{align}
In particular, the chi-squared divergence, squared Hellinger distance and Bhattacharyya distance satisfy
\begin{align}
\label{eq: int. rep. chi^2 div}
\chi^2(P \| Q)
&= \int_0^{\infty} \left(1 - \mathds{F}_{P\|Q}(\log \beta) \right)  \, \mathrm{d}\beta - 1; \\[0.1cm]
\label{eq: int. rep. H^2 dist}
\mathscr{H}^2(P\|Q)
&= 1 - \tfrac12 \int_0^\infty \beta^{-\frac32} \, \mathds{F}_{P\|Q}(\log \beta) \, \mathrm{d}\beta; \\[0.1cm]
\label{eq: int. rep. B dist}
B(P\|Q) &= \log 2 - \log \left( \int_0^\infty \beta^{-\frac32} \, \mathds{F}_{P\|Q}(\log \beta) \, \mathrm{d}\beta \right),
\end{align}
where \eqref{eq: int. rep. chi^2 div} appears in \cite[(439]{ISSV16}.

\item R\'enyi divergence \cite[(426) and (427)]{ISSV16}:
For $\alpha \in (0,1) \cup (1,\infty)$,
\begin{align}  \label{eq: int. rep. RenyiD}
D_{\alpha}(P\|Q) &=
\begin{dcases}
\frac1{\alpha-1} \, \log \left( (1-\alpha) \int_0^\infty \beta^{\alpha-2} \,
\mathds{F}_{P\|Q}(\log \beta) \, \mathrm{d}\beta \right), &\; \alpha \in (0,1) \\[0.2cm]
\frac1{\alpha-1} \, \log \left( (\alpha-1) \int_0^{\infty} \beta^{\alpha-2}
\left(1 - \mathds{F}_{P\|Q}(\log \beta) \right)  \, \mathrm{d}\beta \right), &\; \alpha \in (1, \infty).
\end{dcases}
\end{align}

\item $\chi^s$ divergence: For $s \geq 1$
\begin{align}
\label{eq: int. rep. chi^s}
\chi^s(P\|Q) &= \int_1^\infty \frac1{\beta} \left(s-1+\frac1\beta \right)
(\beta-1)^{s-1} \left(1 - \mathds{F}_{P\|Q}(\log \beta) \right)  \, \mathrm{d}\beta \nonumber \\[0.1cm]
& \hspace*{0.4cm} + \int_0^1 \frac1{\beta} \left(s-1+\frac1\beta \right)
(1-\beta)^{s-1} \, \mathds{F}_{P\|Q}(\log \beta)  \, \mathrm{d}\beta.
\end{align}
In particular, the following identities hold for the total variation distance:
\begin{align}  \label{eq: int. rep. TV}
|P-Q| &= 2 \int_1^{\infty} \frac{1-\mathds{F}_{P\|Q}(\log \beta)}{\beta^2} \, \mathrm{d}\beta \\[0.1cm]
\label{eq2: int. rep. TV}
&= 2 \int_0^1 \frac{\mathds{F}_{P\|Q}(\log \beta)}{\beta^2} \, \mathrm{d}\beta,
\end{align}
where \eqref{eq: int. rep. TV} appears in \cite[(214)]{ISSV16}.

\item DeGroot statistical information:
\begin{align}
\label{eq: int. rep. DeGroot Info}
\set{I}_w(P \| Q) &=
\begin{dcases}
(1-w) \int_0^{\frac{1-w}{w}} \frac{\mathds{F}_{P\|Q}(\log \beta)}{\beta^2}
\, \mathrm{d}\beta, & \; w \in \bigl(\tfrac12, 1) \\[0.2cm]
(1-w) \int_{\frac{1-w}{w}}^{\infty} \frac{1-\mathds{F}_{P\|Q}(\log \beta)}{\beta^2}
\, \mathrm{d}\beta, & \; w \in \bigl(0, \tfrac12\bigr].
\end{dcases}
\end{align}

\item Triangular discrimination:
\begin{align}  \label{eq: int. rep. TD}
\Delta(P\|Q) &= 4 \int_0^\infty \frac{1-\mathds{F}_{P\|Q}(\log \beta)}{(\beta+1)^2} \, \mathrm{d}\beta - 2.
\end{align}

\item Lin's measure: For $\theta \in [0,1]$,
\begin{align} \label{eq: int. rep. Lin's div}
L_\theta(P \| Q) &= h(\theta) - (1-\theta) \int_0^{\infty}
\frac{\log \left(1 + \frac{\theta \beta}{1-\theta}\right)}{\beta^2}
\; \mathds{F}_{P\|Q}(\log \beta) \, \mathrm{d}\beta,
\end{align}
where $h \colon [0,1] \mapsto [0, \log 2]$ denotes the binary entropy function. Specifically, the
Jensen-Shannon divergence admits the integral representation:
\begin{align} \label{eq: int. rep. JS div}
\mathrm{JS}(P\|Q) &= \log 2 - \int_0^\infty \frac{\log(\beta+1)}{2 \beta^2}
\; \mathds{F}_{P\|Q}(\log \beta) \; \mathrm{d}\beta.
\end{align}

\item Jeffrey's divergence:
\begin{align}
J(P\|Q) &= \int_1^\infty \bigl(1 - \mathds{F}_{P\|Q}(\log \beta) \bigr) \left( \frac{\log e}{\beta}
+ \frac{\log \beta}{\beta^2} \right) \, \mathrm{d}\beta \nonumber \\[0.1cm]
\label{eq: int. rep. Jefreey's div}
& \hspace*{0.4cm} - \int_0^1 \mathds{F}_{P\|Q}(\log \beta) \, \left( \frac{\log e}{\beta}
+ \frac{\log \beta}{\beta^2} \right) \, \mathrm{d}\beta.
\end{align}

\item $E_\gamma$ divergence: For $\gamma \geq 1$,
\begin{align}  \label{eq: int. rep. E_gamma}
E_\gamma(P\|Q) = \gamma \int_{\gamma}^{\infty} \frac{1 - \mathds{F}_{P\|Q}(\log \beta)}{\beta^2} \, \mathrm{d}\beta.
\end{align}
\end{enumerate}
\end{theorem}

\begin{proof}
See Appendix~\ref{appendix: proof of identities}.
\end{proof}

\vspace*{0.2cm}
An application of \eqref{eq: int. rep. E_gamma} yields the following interplay
between the $E_\gamma$ divergence and the relative information spectrum.
\begin{theorem} \label{theorem: RIS -- EG}
Let $X \sim P$, and let the random variable $\imath_{P\|Q}(X)$ have no probability masses. Denote
\begin{align}
\label{eq: A1}
& \set{A}_1 := \bigl\{ E_\gamma(P\|Q) : \gamma \geq 1 \bigr\}, \\[0.1cm]
\label{eq: A2}
& \set{A}_2 := \bigl\{ E_\gamma(Q\|P) : \gamma > 1 \bigr\}.
\end{align}
Then,
\begin{enumerate}[a)]
\item
$E_{\gamma}(P\|Q)$ is a continuously differentiable function of $\gamma$ on
$(1, \infty)$, and $E'_{\gamma}(P\|Q) \leq 0$;
\item
the sets $\set{A}_1$ and $\set{A}_2$ determine, respectively, the relative
information spectrum $\mathds{F}_{P\|Q}(\cdot)$ on $[0, \infty)$ and $(-\infty, 0)$;
\item
for $\gamma > 1$,
\begin{align}
\label{eq: RIS - positive arguments}
& \mathds{F}_{P\|Q}(+\log \gamma) = 1 - E_\gamma(P\|Q) + \gamma E'_\gamma(P\|Q), \\
\label{eq: RIS - negative arguments}
& \mathds{F}_{P\|Q}(-\log \gamma) = -E'_\gamma(Q\|P),  \\
\label{eq1: RIS at 0}
& \mathds{F}_{P\|Q}(0) = 1 - E_1(P\|Q) + \lim_{\gamma \downarrow 1} E'_\gamma(P\|Q) \\
\label{eq2: RIS at 0}
& \hspace*{1.4cm} = -\lim_{\gamma \downarrow 1}  E'_\gamma(Q\|P).
\end{align}
\end{enumerate}
\end{theorem}

\begin{proof}
We first prove Item~a).
By our assumption, $\mathds{F}_{P\|Q}(\cdot)$ is continuous on $\Reals$. Hence, it follows from
\eqref{eq: int. rep. E_gamma} that $E_\gamma(P\|Q)$ is continuously differentiable in $\gamma \in (1, \infty)$;
furthermore, \eqref{eq1:E_gamma} implies that $E_\gamma(P\|Q)$ is monotonically decreasing
in $\gamma$, which yields $E'_\gamma(P\|Q) \leq 0$.

We next prove Items~b) and c) together. Let $X \sim P$ and $Y \sim Q$.
From \eqref{eq: int. rep. E_gamma}, for $\gamma > 1$,
\begin{align}
\frac{\mathrm{d}}{\mathrm{d}\gamma} \left(\frac{E_\gamma(P\|Q)}{\gamma}\right)
= -\frac{1-\mathds{F}_{P\|Q}(\log \gamma)}{\gamma^2},
\end{align}
which yields \eqref{eq: RIS - positive arguments}. Due to the continuity of
$\mathds{F}_{P\|Q}(\cdot)$, it follows that the set $\set{A}_1$ determines
the relative information spectrum on $[0, \infty)$.

To prove \eqref{eq: RIS - negative arguments}, we have
\begin{align}
\label{RIS - eq1}
E_{\gamma}(Q\|P) &= \prob[\imath_{Q\|P}(Y) > \log \gamma] - \gamma \, \prob[\imath_{Q\|P}(X) > \log \gamma] \\
\label{RIS - eq2}
&= 1 - \mathds{F}_{Q\|P}(\log \gamma) - \gamma \, \prob[\imath_{Q\|P}(X) > \log \gamma] \\
\label{RIS - eq3}
&= E_\gamma(Q\|P) - \gamma E'_\gamma(Q\|P) - \gamma \, \prob[\imath_{Q\|P}(X) > \log \gamma] \\
\label{RIS - eq4}
&= E_\gamma(Q\|P) - \gamma E'_\gamma(Q\|P) - \gamma \, \prob[\imath_{P\|Q}(X) < -\log \gamma] \\
\label{RIS - eq5}
&= E_\gamma(Q\|P) - \gamma E'_\gamma(Q\|P) - \gamma \, \mathds{F}_{P\|Q}(-\log \gamma)
\end{align}
where \eqref{RIS - eq1} holds by switching $P$ and $Q$ in \eqref{eq2:E_gamma}; \eqref{RIS - eq2}
holds since $Y \sim Q$; \eqref{RIS - eq3} holds by switching $P$ and $Q$ in \eqref{eq: RIS - positive arguments}
(correspondingly, also $X \sim P$ and $Y \sim Q$ are switched); \eqref{RIS - eq4} holds since
$\imath_{Q\|P} = -\imath_{P\|Q}$; \eqref{RIS - eq5} holds by the assumption that
$\frac{\mathrm{d}P}{\mathrm{d}Q} \, (X)$ has no probability masses, which implies that the sign $<$
can be replaced with $\leq$ at the term $\prob[\imath_{P\|Q}(X) < -\log \gamma]$ in the right side
of \eqref{RIS - eq4}. Finally, \eqref{eq: RIS - negative arguments} readily follows from
\eqref{RIS - eq1}--\eqref{RIS - eq5}, which implies that the set $\set{A}_2$ determines
$\mathds{F}_{P\|Q}(\cdot)$ on $(-\infty, 0)$.

Equalities \eqref{eq1: RIS at 0} and \eqref{eq1: RIS at 0} finally follows by letting
$\gamma \downarrow 1$, respectively, on both sides of \eqref{eq: RIS - positive arguments}
and \eqref{eq: RIS - negative arguments}.
\end{proof}

A similar application of \eqref{eq: int. rep. DeGroot Info} yields an interplay
between DeGroot statistical information and the relative information spectrum.
\begin{theorem}
\label{theorem: RIS -- DG}
Let $X \sim P$, and let the random variable $\imath_{P\|Q}(X)$ have no probability masses. Denote
\begin{align}
\label{eq: B1}
& \set{B}_1 := \Bigl\{ \mathcal{I}_\omega(P\|Q) : \omega \in \bigl(0, \tfrac12 \bigr] \Bigr\}, \\[0.1cm]
\label{eq: B2}
& \set{B}_2 := \Bigl\{ \mathcal{I}_\omega(P\|Q) : \omega \in \bigl(\tfrac12, 1 \bigr) \Bigr\}.
\end{align}
Then,
\begin{enumerate}[a)]
\item
$\mathcal{I}_\omega(P\|Q)$ is a continuously differentiable function of $\omega$
on $(0, \tfrac12) \cup (\tfrac12, 1)$,
\begin{align}
\lim_{\omega \uparrow \tfrac12} \, \mathcal{I}'_\omega(P\|Q) - \lim_{\omega \downarrow \tfrac12} \, \mathcal{I}'_\omega(P\|Q) = 2,
\end{align}
and $\mathcal{I}'_\omega(P\|Q)$ is, respectively, non-negative or non-positive on $\bigl(0, \tfrac12\bigr)$ and $\bigl(\tfrac12, 1 \bigr)$;
\item
the sets $\set{B}_1$ and $\set{B}_2$ determine, respectively, the relative
information spectrum $\mathds{F}_{P\|Q}(\cdot)$ on $[0, \infty)$ and $(-\infty, 0)$;
\item
for $\omega \in \bigl(0, \tfrac12 \bigr)$
\begin{align}
\label{eq2: RIS - pos. arg.}
\mathds{F}_{P\|Q}\left(\log \tfrac{1-\omega}{\omega}\right)
= 1 - \mathcal{I}_\omega(P\|Q) - (1-\omega) \, \mathcal{I}'_\omega(P\|Q),
\end{align}
for $\omega \in \bigl(\tfrac12, 1\bigr)$
\begin{align}
\label{eq2: RIS - neg. arg.}
\mathds{F}_{P\|Q}\left(\log \tfrac{1-\omega}{\omega}\right)
= - \mathcal{I}_\omega(P\|Q) - (1-\omega) \, \mathcal{I}'_\omega(P\|Q),
\end{align}
and
\begin{align}
\label{eq3: RIS at 0}
\mathds{F}_{P\|Q}(0) &=
-\mathcal{I}_{\frac12}(P\|Q) - \tfrac12 \lim_{\omega \downarrow \tfrac12}
\, \mathcal{I}'_\omega(P\|Q).
\end{align}
\end{enumerate}
\end{theorem}

\begin{remark} \label{remark: discontinuities}
By relaxing the condition in Theorems~\ref{theorem: RIS -- EG} and \ref{theorem: RIS -- DG}
where $\frac{\mathrm{d}P}{\mathrm{d}Q} \, (X)$ has no probability masses with $X \sim P$,
it follows from the proof of Theorem~\ref{theorem: RIS -- EG} that each one of the sets
\begin{align}
\label{eq: union A}
& \set{A} := \set{A}_1 \cup \set{A}_2 = \Bigl\{ \bigl(E_\gamma(P\|Q), E_\gamma(Q\|P) \bigr): \gamma \geq 1 \Bigr\}, \\
\label{eq: union B}
& \set{B} := \set{B}_1 \cup \set{B}_2 = \Bigl\{ \mathcal{I}_\omega(P\|Q): \omega \in (0,1) \Bigr\}
\end{align}
determines $\mathds{F}_{P\|Q}(\cdot)$ at every point on $\Reals$ where this
relative information spectrum is continuous. Note that, as a cumulative distribution
function, $\mathds{F}_{P\|Q}(\cdot)$ is discontinuous at a countable number of points.
Consequently, under the condition that $f \in \set{C}$ is differentiable on $(0, \infty)$,
the integral representations of $D_f(P\|Q)$ in Theorem~\ref{theorem: Int. rep.} are not
affected by the countable number of discontinuities for $\mathds{F}_{P\|Q}(\cdot)$.
\end{remark}

In view of Theorems~\ref{theorem: Int. rep.}, \ref{theorem: RIS -- EG} and \ref{theorem: RIS -- DG}
and Remark~\ref{remark: discontinuities}, we get the following result.
\begin{corollary} \label{corollary: DG-EG/ D_f}
Let $f \in \set{C}$ be a differentiable function on $(0, \infty)$, and let $P \ll \gg Q$
be probability measures. Then, each one of the sets $\set{A}$ and $\set{B}$ in \eqref{eq: union A}
and \eqref{eq: union B}, respectively, determines $D_f(P\|Q)$.
\end{corollary}

\begin{remark}
Corollary~\ref{corollary: DG-EG/ D_f} is supported by the integral representation of $D_f(P\|Q)$
in \cite[Theorem~11]{LieseV_IT2006}, expressed as a function of the set of values in $\set{B}$,
and its analogous representation in \cite[Proposition~3]{ISSV16} as a function of the set of values
in $\set{A}$. More explicitly, \cite[Theorem~11]{LieseV_IT2006} states that if $f \in \set{C}$, then
\begin{align} \label{eq1: LieseV06}
D_f(P\|Q) = \int_0^1 \set{I}_\omega(P\|Q) \, \text{d}\Gamma_f(\omega)
\end{align}
where $\Gamma_f$ is a certain $\sigma$-finite measure defined on the Borel subsets of $(0,1)$;
it is also shown in \cite[(80)]{LieseV_IT2006} that if $f \in \set{C}$ is twice
differentiable on $(0, \infty)$, then
\begin{align}
\label{eq4: LieseV06}
D_f(P\|Q)
&= \int_0^1  \set{I}_\omega(P\|Q) \; \frac1{\omega^3}
\; f''\left(\frac{\omega}{1-\omega}\right) \, \text{d}\omega.
\end{align}
\end{remark}

\vspace*{0.2cm}
\section{New $f$-divergence Inequalities}
\label{section: new inequalities}

Various approaches for the derivation of $f$-divergence inequalities were studied
in the literature (see Section~\ref{section: introduction} for references). This section suggests
a new approach, leading to a lower bound on an arbitrary $f$-divergence by
means of the $E_\gamma$ divergence of an arbitrary order $\gamma \geq 1$
(see \eqref{eq1:E_gamma}) or the DeGroot statistical information (see \eqref{eq:DG f-div}).
This approach leads to generalizations of the Bretagnole-Huber inequality
\cite{BretagnolleH79}, whose generalizations are later motivated in this section.
The utility of the $f$-divergence inequalities
in this section is exemplified in the setup of Bayesian binary hypothesis testing.

In the following, we provide the first main result in this section for the
derivation of new $f$-divergence inequalities by means of the $E_\gamma$
divergence. Generalizing the total variation distance, the $E_\gamma$ divergence
in \eqref{eq1:E_gamma}--\eqref{eq:Eg f-div} is an $f$-divergence whose
utility in information theory has been exemplified in \cite[Chapter~3]{CZK98},
\cite{Liu17}, \cite[p.~2314]{PPV10} and \cite{PW15}; the properties of
this measure were studied in \cite[Section~2.B]{Liu17} and
\cite[Section~7]{ISSV16}.

\vspace*{0.1cm}
\begin{theorem} \label{theorem: f-div ineq}
Let $f \in \set{C}$, and let $f^\ast \in \set{C}$ be the
conjugate convex function as defined in \eqref{eq: conjugate f}.
Let $P$ and $Q$ be probability measures.
Then, for all $\gamma \in [1, \infty)$,
\begin{align} \label{eq: f-div ineq}
D_f(P\|Q) \geq f^\ast\left(1 + \tfrac1\gamma \, E_\gamma(P\|Q) \right)
+ f^\ast\left(\tfrac1\gamma \, \bigl(1- E_\gamma(P\|Q) \bigr) \right)
- f^\ast\left( \tfrac1\gamma \right).
\end{align}
\end{theorem}

\begin{proof}
Let $p = \frac{\mathrm{d}P}{\mathrm{d}\mu}$
and $q = \frac{\mathrm{d}Q}{\mathrm{d}\mu}$
be the densities of $P$ and $Q$
with respect to a dominating measure
$\mu$ $(P, Q \ll \mu)$. Then,
for an arbitrary $a \in \Reals$,
\begin{align}
D_f(P \| Q) &= D_{f^\ast}(Q \| P) \\[0.1cm]
&= \int p \, f^\ast\left( \frac{q}{p} \right) \, \mathrm{d}\mu \\[0.1cm]
&= \int p \left[ f^\ast\left( \max\left\{a, \frac{q}{p} \right\} \right)
+ f^\ast\left( \min\left\{a, \frac{q}{p} \right\} \right) - f^\ast(a) \right] \, \mathrm{d}\mu \\[0.1cm]
\label{eq1: Jensen}
&\geq f^\ast\left( \int p \max\left\{a, \frac{q}{p} \right\} \, \mathrm{d}\mu \right)
+ f^\ast\left( \int p \min\left\{a, \frac{q}{p} \right\} \, \mathrm{d}\mu \right) - f^\ast(a)
\end{align}
where \eqref{eq1: Jensen} follows from the convexity of $f^\ast$ and by invoking Jensen's inequality.

Setting $a := \frac1\gamma$ with $\gamma \in [1, \infty)$ gives
\begin{align}
\int p \max\left\{a, \frac{q}{p} \right\} \, \mathrm{d}\mu
&= \int \max \left\{ \frac{p}{\gamma}, q \right\} \, \mathrm{d}\mu \\[0.1cm]
&= \int q \, \mathrm{d}\mu + \int \max \left\{ \frac{p}{\gamma} - q, 0 \right\} \, \mathrm{d}\mu \\[0.1cm]
&= 1 + \frac1\gamma \int q \max\left\{\frac{p}{q} - \gamma, 0 \right\} \, \mathrm{d}\mu \\[0.1cm]
\label{eq: 1st integral}
&= 1 + \tfrac1\gamma \, E_\gamma(P \| Q),
\end{align}
and
\begin{align}
\int p \min\left\{a, \frac{q}{p} \right\} \, \mathrm{d}\mu
&= \int p \left( a + \frac{q}{p} - \max\left\{a, \frac{q}{p} \right\} \right) \mathrm{d}\mu \\[0.1cm]
&= a+1 - \int p \max\left\{a, \frac{q}{p} \right\} \, \mathrm{d}\mu \\[0.1cm]
\label{eq: 2nd integral}
&= \tfrac1\gamma \, \bigl( 1 - E_\gamma(P \| Q) \bigr)
\end{align}
where \eqref{eq: 2nd integral} follows from \eqref{eq: 1st integral} by setting $a := \frac1\gamma$.
Substituting \eqref{eq: 1st integral} and \eqref{eq: 2nd integral} into the right side of
\eqref{eq1: Jensen} gives \eqref{eq: f-div ineq}.
\end{proof}

An application of Theorem~\ref{theorem: f-div ineq} gives the following lower bounds
on the Hellinger and R\'enyi divergences with arbitrary positive orders, expressed
as a function of the $E_\gamma$ divergence with an arbitrary order $\gamma \geq 1$.

\begin{corollary} \label{cor: Hel RD}
For all $\alpha > 0$ and $\gamma \geq 1$,
\begin{align} \label{eq: cor - Hel}
\mathscr{H}_{\alpha}(P \| Q) \geq
\begin{dcases}
\frac1{\alpha-1} \left[ \left(1 + \frac1\gamma \, E_\gamma(P\|Q) \right)^{1-\alpha}
+ \left( \frac{1-E_\gamma(P\|Q)}{\gamma} \right)^{1-\alpha} - 1 -
\gamma^{\alpha-1} \right], & \quad \alpha \neq 1 \\[0.2cm]
-\log_{\mathrm{e}} \Biggl( \left(1 + \frac1\gamma \, E_\gamma(P\|Q) \right)
\bigl(1-E_\gamma(P\|Q)\bigr) \Biggr), &\quad \alpha=1,
\end{dcases}
\end{align}
and
\begin{align} \label{eq: cor - RD}
D_{\alpha}(P \| Q) \geq
\begin{dcases}
\frac1{\alpha-1} \log \Biggl( \left(1 + \frac1\gamma \, E_\gamma(P\|Q) \right)^{1-\alpha}
+ \gamma^{\alpha-1} \left[ \bigl(1-E_\gamma(P\|Q) \bigr)^{1-\alpha} - 1 \right] \Biggr),
& \quad \alpha \neq 1 \\[0.2cm]
-\log \Biggl( \left(1 + \frac1\gamma \, E_\gamma(P\|Q) \right) \bigl(1-E_\gamma(P\|Q)\bigr)
\Biggr), &\quad \alpha=1.
\end{dcases}
\end{align}
\end{corollary}
\begin{proof}
Inequality~\eqref{eq: cor - Hel}, for $\alpha \in (0,1) \cup (1, \infty)$, follows from
Theorem~\ref{theorem: f-div ineq} and \eqref{eq: Hel-divergence}; for $\alpha=1$,
it holds in view of Theorem~\ref{theorem: f-div ineq}, and equalities \eqref{eq: KL divergence}
and \eqref{eq1: KL}. Inequality~\eqref{eq: cor - RD}, for $\alpha \in (0,1) \cup (1, \infty)$,
follows from \eqref{renyimeetshellinger} and \eqref{eq: cor - Hel}; for $\alpha=1$,
it holds in view of \eqref{eq1: KL}, \eqref{eq: cor - Hel} and since $D_1(P\|Q)=D(P\|Q)$.
\end{proof}

\vspace*{0.1cm}
Specialization of Corollary~\ref{cor: Hel RD} for $\alpha=2$ in \eqref{eq: cor - Hel} and $\alpha=1$
in \eqref{eq: cor - RD} gives the following result.
\begin{corollary} \label{corollary 3}
For $\gamma \in [1, \infty)$, the following upper bounds on $E_\gamma$ divergence hold as a function
of the relative entropy and $\chi^2$ divergence:
\begin{align}
\label{eq: chi^2-EG}
& E_\gamma(P \| Q) \leq \tfrac12 \left[ 1-\gamma + \sqrt{(\gamma-1)^2
+ \frac{4 \gamma \, \chi^2(P \| Q)}{1 + \gamma + \chi^2(P \| Q)}} \; \right], \\[0.2cm]
\label{eq: generalized BH}
& E_\gamma(P \| Q) \leq \tfrac12 \left[ 1-\gamma + \sqrt{(\gamma-1)^2
+ 4 \gamma \bigl( 1 -\exp(-D(P\|Q)) \bigr)} \right].
\end{align}
\end{corollary}

\begin{remark}
From \cite[(58)]{ReidW11},
\begin{align} \label{eq: lb chi-square - TV}
\chi^2(P\|Q) \geq
\begin{dcases}
|P-Q|^2, & \quad \mbox{$|P-Q| \in \bigl[0, 1)$} \\[0.2cm]
\frac{|P-Q|}{2-|P-Q|}, & \quad \mbox{$|P-Q| \in \bigl[1, 2)$}
\end{dcases}
\end{align}
is a tight lower bound on the chi-squared divergence as a function of the
total variation distance. In view of \eqref{eq:EG-TV}, we compare
\eqref{eq: lb chi-square - TV} with the specialized version of
\eqref{eq: chi^2-EG} when $\gamma=1$. The latter bound is expected to be
looser than the tight bound in \eqref{eq: lb chi-square - TV}, as a result of
the use of Jensen's inequality in the proof of Theorem~\ref{theorem: f-div ineq};
however, it is interesting to examine how much we loose in the tightness of
this specialized bound with $\gamma=1$. From \eqref{eq:EG-TV}, the substitution
of $\gamma=1$ in \eqref{eq: chi^2-EG} gives
\begin{align} \label{eq2: LB Chi^2-TV}
\chi^2(P\|Q) \geq \frac{2 |P-Q|^2}{4-|P-Q|^2}, & \qquad |P-Q| \in [0,2),
\end{align}
and, it can be easily verified that
\begin{itemize}
\item if $|P-Q| \in [0,1)$, then the lower bound in the right side of \eqref{eq2: LB Chi^2-TV} is
at most twice smaller than the tight lower bound in the right side of
\eqref{eq: lb chi-square - TV};
\item if $|P-Q| \in [1,2)$, then the lower bound in the right side of \eqref{eq2: LB Chi^2-TV} is
at most $\tfrac{3}{2}$ times smaller than the tight lower bound in the right side of
\eqref{eq: lb chi-square - TV}.
\end{itemize}
\end{remark}

\begin{remark}
Setting $\gamma=1$ in \eqref{eq: generalized BH}, and using \eqref{eq:EG-TV},
specializes to the Bretagnole-Huber inequality \cite{BretagnolleH79}:
\begin{align} \label{eq: BretagnolleH79}
|P-Q| \leq 2 \sqrt{1- \exp\bigl(-D(P\|Q)\bigr)}.
\end{align}
\end{remark}

\par
Inequality \eqref{eq: BretagnolleH79} forms a counterpart to Pinsker's inequality:
\begin{align}  \label{eq: Pinsker}
\tfrac12  |P-Q|^2 \log e \leq D(P \| Q),
\end{align}
proved by Csisz\'{a}r \cite{Csiszar67a} and Kullback \cite{kullbackTV67}, with Kemperman
\cite{kemperman} independently a bit later. As upper bounds on the total variation distance,
\eqref{eq: Pinsker} outperforms \eqref{eq: BretagnolleH79} if $D(P\|Q) \leq 1.594$ nats,
and \eqref{eq: BretagnolleH79} outperforms \eqref{eq: Pinsker} for larger values of
$D(P\|Q)$.

\begin{remark}
In \cite[(8)]{Vajda70}, Vajda introduced a lower bound on the relative entropy
as a function of the total variation distance:
\begin{align} \label{eq: Vajda's LB}
D(P\|Q) \geq \log\left(\frac{2+|P-Q|}{2-|P-Q|}\right)-\frac{2|P-Q| \, \log e}{2+|P-Q|}, \quad |P-Q| \in [0,2).
\end{align}
The lower bound in the right side of \eqref{eq: Vajda's LB} is asymptotically tight
in the sense that it tends to $\infty$ if $|P-Q| \uparrow 2$, and the difference between
$D(P\|Q)$ and this lower bound is everywhere upper bounded
by $\frac{2|P-Q|^3}{(2+|P-Q|)^2} \leq 4$ (see \cite[(9)]{Vajda70}). The Bretagnole-Huber
inequality in \eqref{eq: BretagnolleH79}, on the other hand, is equivalent to
\begin{align} \label{eq2: BretagnolleH79}
D(P\|Q) \geq -\log\left(1 - \tfrac14 |P-Q|^2 \right), \quad |P-Q| \in [0,2).
\end{align}
Although it can be verified numerically that the lower bound on the relative entropy
in \eqref{eq: Vajda's LB} is everywhere slightly tighter than the lower bound in \eqref{eq2: BretagnolleH79}
(for $|P-Q| \in [0,2)$), both lower bounds on $D(P\|Q)$ are of the same asymptotic tightness
in a sense that they both tend to $\infty$ as $|P-Q| \uparrow 2$ and their ratio tends to~1.
Apart of their asymptotic tightness, the Bretagnole-Huber inequality in \eqref{eq2: BretagnolleH79}
is appealing since it provides a closed-form simple upper bound on $|P-Q|$ as a function of $D(P\|Q)$
(see \eqref{eq: BretagnolleH79}), whereas such a closed-form simple upper bound cannot be obtained from
\eqref{eq: Vajda's LB}.
In fact, by the substitution $v := -\frac{2-|P-Q|}{2+|P-Q|}$ and the exponentiation of both sides
of \eqref{eq: Vajda's LB}, we get the inequality $v e^v \geq -\tfrac1e \, \exp\bigl(-D(P\|Q)\bigr)$
whose solution is expressed by the Lambert $W$ function \cite{Corless96};
it can be verified that \eqref{eq: Vajda's LB} is equivalent to the following upper bound
on the total variation distance as a function of the relative entropy:
\begin{align}
\label{eq: UB igal}
& |P-Q| \leq \frac{2 \bigl(1+W(z)\bigr)}{1-W(z)}, \\[0.1cm]
\label{eq2: UB igal}
& z := -\tfrac1{e} \, \exp\bigl(-D(P\|Q)\bigr),
\end{align}
where $W$ in the right side of \eqref{eq: UB igal} denotes the principal real branch of the Lambert $W$ function.
The difference between the upper bounds in \eqref{eq: BretagnolleH79} and \eqref{eq: UB igal}
can be verified to be marginal if $D(P\|Q)$ is large (e.g., if $D(P\|Q)=4$ nats, then the upper bounds
on $|P-Q|$ are respectively equal to 1.982 and 1.973), though the former upper bound in \eqref{eq: BretagnolleH79}
is clearly more simple and amenable to analysis.

The Bretagnole-Huber inequality in \eqref{eq: BretagnolleH79} is proved to be useful in the context of lower
bounding the minimax risk (see, e.g., \cite[pp.~89--90, 94]{Tsybakov09}), and the problem of density estimation
(see, e.g., \cite[Section~1.6]{Vapnik98}). The utility of this inequality motivates its generalization in this section
(see Corollaries~\ref{cor: Hel RD} and~\ref{corollary 3}, and also see later Theorem~\ref{theorem: DG UBs} followed by
Example~\ref{example: Poisson}).
\end{remark}

\vspace*{0.2cm}
In \cite[Section~7.C]{ISSV16}, Sason and Verd\'{u} generalized Pinsker's inequality
by providing an upper bound on the $E_\gamma$ divergence, for $\gamma>1$, as a function
of the relative entropy. In view of \eqref{eq:EG-TV} and the optimality of the constant in
Pinsker's inequality \eqref{eq: Pinsker}, it follows
that the minimum achievable $D(P\|Q)$ is quadratic in $E_1(P\|Q)$ for small values of $E_1 (P\|Q)$.
It has been proved in \cite[Section~7.C]{ISSV16} that this situation ceases to be the case for
$\gamma > 1$, in which case it is possible to upper bound $E_\gamma(P\|Q)$ as a constant times
$D(P\|Q)$ where this constant tends to infinity as we let $\gamma \downarrow 1$.
We next cite the result in \cite[Theorem~30]{ISSV16}, extending \eqref{eq: Pinsker}
by means of the $E_\gamma$ divergence for $\gamma>1$, and compare it numerically to
the bound in \eqref{eq: generalized BH}.

\begin{theorem} (\cite[Theorem~30]{ISSV16})
\label{thm:EG vs. RE}
For every $\gamma > 1$,
\begin{align}  \label{eq:sup-EG and RE}
\sup \frac{E_{\gamma}(P\|Q)}{D(P\|Q)} = c_{\gamma}
\end{align}
where the supremum is over $P \ll Q,  P \neq Q$, and $c_{\gamma}$ is a universal
function (independent of $(P, Q)$), given by
\begin{align}
\label{eq: c_gamma}
& c_{\gamma} = \frac{t_\gamma-\gamma}{t_\gamma \, \log t_\gamma + (1-t_\gamma) \, \log e},
\\[0.1cm]
\label{eq: t_gamma}
& t_\gamma = - \gamma \, W_{-1}\left(-\tfrac1{\gamma} \, e^{-\frac1{\gamma}} \right)
\end{align}
where $W_{-1}$ in \eqref{eq: t_gamma} denotes the secondary real branch of the
Lambert $W$ function \cite{Corless96}.
\end{theorem}

As an immediate consequence of \eqref{eq:sup-EG and RE}, it follows that
\begin{align} \label{eq:SLB-EG-RE}
E_{\gamma}(P\|Q) \leq  c_{\gamma} D(P\|Q),
\end{align}
which forms a straight-line bound on the $E_{\gamma}$ divergence as a function of
the relative entropy for $\gamma>1$.
Similarly to the comparison of the Bretagnole-Huber inequality \eqref{eq: BretagnolleH79} and
Pinsker's inequality \eqref{eq: Pinsker}, we exemplify numerically that the extension of
Pinsker's inequality to the $E_\gamma$ divergence in \eqref{eq:SLB-EG-RE} forms a counterpart
to the generalized version of the Bretagnole-Huber inequality in \eqref{eq: generalized BH}.

\begin{figure}[h]
\vspace*{-4.5cm}
\centerline{\includegraphics[width=12cm]{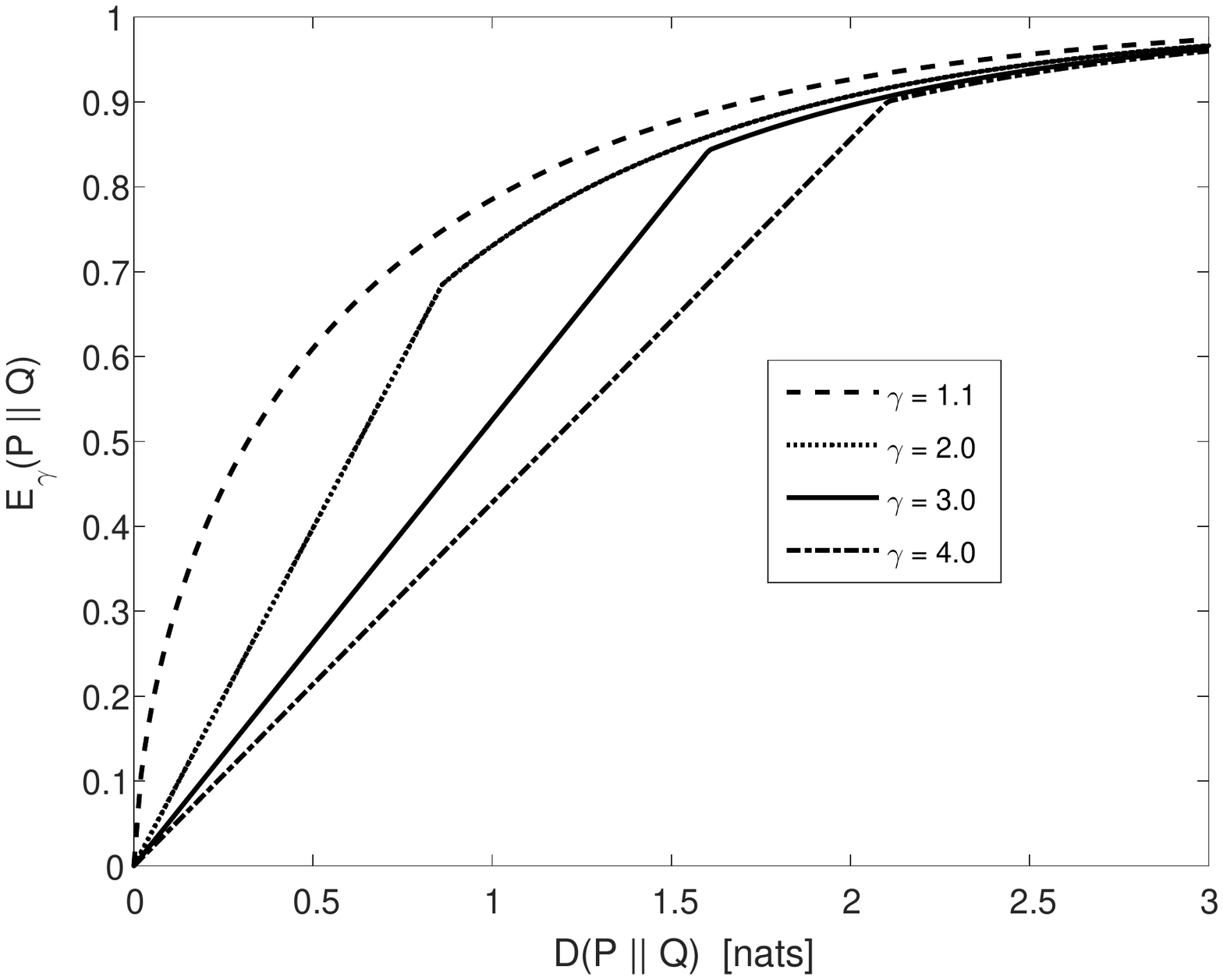}}
\vspace*{-3.6cm}
\caption{\label{figure:RE-EG}
Upper bounds on the $E_\gamma$ divergence, for $\gamma > 1$, as a function of the relative entropy
(the curvy and straight lines follow from \eqref{eq: generalized BH} and \eqref{eq:SLB-EG-RE}, respectively).}
\end{figure}
Figure~\ref{figure:RE-EG} plots an upper bound on the $E_\gamma$ divergence,
for $\gamma \in \{1.1, 2.0, 3.0, 4.0\}$, as a function of the relative entropy
(or, alternatively, a lower bound on the relative entropy as a function of the
$E_\gamma$ divergence).
The upper bound on $E_\gamma(P\|Q)$ for $\gamma > 1$, as a function of $D(P\|Q)$,
is composed of the following two components:
\begin{enumerate}[a)]
\item the straight-line bound, which refers to the right side of \eqref{eq:SLB-EG-RE}, is tighter than
the bound in the right side of \eqref{eq: generalized BH} if the relative entropy is below a certain value
that is denoted by $d(\gamma)$ in nats (it depends on $\gamma$);
\item the curvy line, which refers to the bound in the right side of \eqref{eq: generalized BH},
is tighter than the straight-line bound in the right side of \eqref{eq:SLB-EG-RE} for larger values
of the relative entropy.
\end{enumerate}
It is supported by Figure~\ref{figure:RE-EG} that $d \colon (1, \infty) \mapsto (0, \infty)$ is positive
and monotonically increasing, and $\underset{\gamma \downarrow 1}{\lim} \, d(\gamma)=0$; e.g., it
can be verified that $d(1.1) \approx 0.02$, $d(2) \approx 0.86$, $d(3) \approx 1.61$, and
$d(4) \approx 2.10$ (see Figure~\ref{figure:RE-EG}).

\subsection*{Bayesian Binary Hypothesis Testing}
The DeGroot statistical information \cite{DeGroot62} has the following meaning:
consider two hypotheses $H_{\mathtt{0}}$ and $H_{\mathtt{1}}$, and let
$\prob[H_{\mathtt{0}}]=\omega$ and $\prob[H_{\mathtt{1}}]=1-\omega$ with $\omega \in (0,1)$.
Let $P$ and $Q$ be probability measures, and consider an observation $Y$ where
$Y | H_{\mathtt{0}} \sim P$, and $Y | H_{\mathtt{1}} \sim Q$. Suppose that one wishes
to decide which hypothesis is more likely given the observation $Y$. The operational meaning
of the DeGroot statistical information, denoted by $\set{I}_\omega(P\|Q)$, is that this
measure is equal to the minimal difference between the {\em a-priori} error probability
(without side information) and {\em a posteriori} error probability (given the observation $Y$).
This measure was later identified as an $f$-divergence by Liese and Vajda \cite{LieseV_IT2006}
(see \eqref{eq:DG f-div} here).

\begin{theorem} \label{theorem: DG UBs}
The DeGroot statistical information satisfies the following upper bound as a function of the
chi-squared divergence:
\begin{align}
\label{eq: DG-chi^2 UB}
\mathcal{I}_\omega(P\|Q) \leq
\begin{dcases}
\omega - \tfrac12 + \sqrt{ \tfrac14 - \frac{\omega (1-\omega)}{1+\omega \, \chi^2(P\|Q)}} \, ,
& \quad \omega \in \bigl(0, \tfrac12 \bigr], \\[0.2cm]
\tfrac12 - \omega + \sqrt{ \tfrac14 - \frac{\omega (1-\omega)}{1+\omega \, \chi^2(Q\|P)}} \, ,
& \quad \omega \in \bigl(\tfrac12, 1\bigr),
\end{dcases}
\end{align}
and the following bounds as a function of the relative entropy:
\begin{enumerate}[1)]
\item
\begin{align}
\label{eq1: DG-KL UB}
\mathcal{I}_\omega(P\|Q) \leq
\begin{dcases}
\omega \, c_{\frac{1-\omega}{\omega}} \, D(P\|Q)\, ,
& \quad \omega \in \bigl(0, \tfrac12\bigr), \\[0.2cm]
\sqrt{ \tfrac1{8 \log e} \, \min\bigl\{ D(P\|Q), D(Q\|P) \bigr\}} \, , & \quad \omega = \tfrac12, \\[0.2cm]
(1-\omega) \, c_{\frac{\omega}{1-\omega}} \, D(Q\|P) \, ,
& \quad \omega \in \bigl(\tfrac12, 1\bigr),
\end{dcases}
\end{align}
where $c_\gamma$ for $\gamma > 1$ is introduced in \eqref{eq: c_gamma};
\item
\begin{align}
\label{eq2: DG-KL UB}
\mathcal{I}_\omega(P\|Q) \leq
\begin{dcases}
\omega - \tfrac12 + \sqrt{\tfrac14 - \omega (1-\omega) \, \exp\bigl(-D(P\|Q)\bigr)}\, ,
& \quad \omega \in \bigl(0, \tfrac12\bigr], \\[0.2cm]
\tfrac12 - \omega + \sqrt{\tfrac14 - \omega (1-\omega) \, \exp\bigl(-D(Q\|P)\bigr)}\, ,
& \quad \omega \in \bigl(\tfrac12, 1\bigr).
\end{dcases}
\end{align}
\end{enumerate}
\end{theorem}

\begin{proof}
The first bound in \eqref{eq: DG-chi^2 UB} holds by combining \eqref{eq: DG-EG} and \eqref{eq: chi^2-EG};
the second bound in \eqref{eq1: DG-KL UB} follows from \eqref{eq:SLB-EG-RE} and \eqref{eq: DG-EG} for
$\omega \in \bigl(0, \tfrac12\bigr) \cup \bigl(\tfrac12, 1\bigr)$, and it follows from
\eqref{eq:DG-TV} and \eqref{eq: Pinsker} when $\omega = \tfrac12$; finally,
the third bound in \eqref{eq2: DG-KL UB} follows from \eqref{eq: generalized BH} and \eqref{eq: DG-EG}.
\end{proof}

\begin{remark}
The bound in \eqref{eq1: DG-KL UB} forms an extension of Pinsker's inequality \eqref{eq: Pinsker}
when $\omega \neq \tfrac12$ (i.e., in the asymmetric case where the hypotheses $H_{\mathtt{0}}$ and
$H_{\mathtt{1}}$ are not equally probable).
Furthermore, in view of \eqref{eq:DG-TV}, the bound in \eqref{eq2: DG-KL UB} is specialized
to the Bretagnole-Huber inequality in \eqref{eq: BretagnolleH79} by letting $\omega = \tfrac12$.
\end{remark}

\begin{remark}
Numerical evidence shows that none of the bounds in \eqref{eq: DG-chi^2 UB}--\eqref{eq2: DG-KL UB}
supersedes the others.
\end{remark}

\begin{remark}
The upper bounds on $\set{I}_\omega(P_\mu \| P_\lambda)$ in \eqref{eq: DG-chi^2 UB} and
\eqref{eq2: DG-KL UB} are asymptotically tight when we let $D(P\|Q)$ and $D(Q\|P)$ tend
to infinity. To verify this, first note that (see \cite[Theorem~5]{GibbsSu02})
\begin{align} \label{grout425 - introduction}
D ( P \| Q) &\leq \log \bigl( 1 + \chi^2(P\| Q) \bigr),
\end{align}
which implies that also $\chi^2(P\| Q)$ and $\chi^2(Q\| P)$ tend to infinity. In this case,
it can be readily verified that the bounds in \eqref{eq: DG-chi^2 UB} and \eqref{eq2: DG-KL UB} are
specialized to $\set{I}_\omega(P\|Q) \leq \min\{\omega, 1-\omega\}$; this upper bound, which
is equal to the {\em a-priori} error probability, is also equal to the DeGroot statistical
information since the {\em a-posterior} error probability tends to zero in the considered
extreme case where $P$ and $Q$ are sufficiently far from each other, so that $H_{\mathtt{0}}$
and $H_{\mathtt{1}}$ are easily distinguishable in high probability when the observation $Y$
is available.
\end{remark}

\begin{remark}
Due to the one-to-one correspondence between the $E_\gamma$ divergence and DeGroot statistical
information in \eqref{eq: DG-EG}, which shows that the two measures are related by a multiplicative
scaling factor, the numerical results shown in Figure~\ref{figure:RE-EG} also apply to the bounds
in \eqref{eq1: DG-KL UB} and \eqref{eq2: DG-KL UB}; i.e., for $\omega \neq \tfrac12$, the first bound
in \eqref{eq1: DG-KL UB} is tighter than the second bound in \eqref{eq2: DG-KL UB} for small values
of the relative entropy, whereas \eqref{eq2: DG-KL UB} becomes tighter than \eqref{eq1: DG-KL UB}
for larger values of the relative entropy.
\end{remark}

\begin{corollary} \label{corolary: DG}
Let $f \in \set{C}$, and let $f^\ast \in \set{C}$ be as defined in \eqref{eq: conjugate f}.
Then,
\begin{enumerate}[1)]
\item
for $w \in (0, \tfrac12\bigr]$,
\begin{align} \label{eq1: ineq. with DG}
D_f(P\|Q) \geq f^\ast\left( 1 + \frac{\set{I}_w(P \| Q)}{1-w} \right)
+ f^\ast\left( \frac{w - \set{I}_w(P \| Q)}{1-w} \right)
- f^\ast\left(\frac{w}{1-w}\right);
\end{align}
\item
for $w \in \bigl(\tfrac12, 1 \bigr)$,
\begin{align} \label{eq2: ineq. with DG}
D_f(P\|Q) \geq f^\ast\left( 1 + \frac{\set{I}_w(Q \| P)}{w} \right)
+ f^\ast\left( \frac{1-w - \set{I}_w(Q \| P)}{w} \right)
- f^\ast\left(\frac{1-w}{w}\right).
\end{align}
\end{enumerate}
\end{corollary}

\begin{proof}
Inequalities~\eqref{eq1: ineq. with DG} and \eqref{eq2: ineq. with DG} follow by combining \eqref{eq: f-div ineq} and \eqref{eq: DG-EG}.
\end{proof}

We end this section by exemplifying the utility of the bounds in Theorem~\ref{theorem: DG UBs}.
\begin{example} \label{example: Poisson}
Let $\prob[H_{\mathtt{0}}]=\omega$ and $\prob[H_{\mathtt{1}}]=1-\omega$ with $\omega \in (0,1)$,
and assume that the observation $Y$ given that the hypothesis is $H_{\mathtt{0}}$ or
$H_{\mathtt{1}}$ is Poisson distributed with the positive parameter $\mu$ or $\lambda$, respectively:
\begin{align}
& Y | H_{\mathtt{0}} \sim P_\mu, \\
& Y | H_{\mathtt{1}} \sim P_\lambda
\end{align}
where
\begin{align}
P_\lambda[k] = \frac{e^{-\lambda} \lambda^k}{k!}, \quad k \in \{0, 1, \ldots\}.
\end{align}
Without any loss of generality, let $\omega \in \bigl(0, \tfrac12\bigr]$.
The bounds on the DeGroot statistical information $\set{I}_\omega(P_\mu \| P_\lambda)$
in Theorem~\ref{theorem: DG UBs} can be expressed in a closed form by relying on the
following identities:
\begin{align}
& D(P_\mu \| P_\lambda) = \mu \log\Bigl(\frac{\mu}{\lambda}\Bigr) + (\lambda-\mu) \log e, \\
& \chi^2(P_\mu \| P_\lambda) = e^{\frac{(\mu-\lambda)^2}{\lambda}} - 1.
\end{align}
In this example, we compare the simple closed-form bounds on $\set{I}_\omega(P_\mu \| P_\lambda)$
in \eqref{eq: DG-chi^2 UB}--\eqref{eq2: DG-KL UB} with its exact value
\begin{align} \label{eq1: exact DG}
\set{I}_\omega(P_\mu \| P_\lambda) &= \min\{\omega, 1-\omega\} - \sum_{k=0}^{\infty} \min \Bigl\{ \omega P_\mu[k], (1-\omega) P_\lambda[k] \Bigr\}.
\end{align}
To simplify the right side of \eqref{eq1: exact DG}, let $\mu > \lambda$, and define
\begin{align} \label{eq: k_0}
k_0 = k_0(\lambda, \mu, \omega) := \left\lfloor \frac{\mu-\lambda
+ \ln \frac{1-\omega}{\omega}}{\ln \frac{\mu}{\lambda}} \right\rfloor,
\end{align}
where, for $x \in \Reals$, $\lfloor x \rfloor$ denotes the largest integer
that is smaller than or equal to $x$. It can be verified that
\begin{align}
\begin{dcases} \label{eq: compare PMFs}
\omega P_\mu[k] \leq (1-\omega) P_\lambda[k], & \quad \mbox{for $k \leq k_0$}\\
\omega P_\mu[k] > (1-\omega) P_\lambda[k], & \quad \mbox{for $k > k_0$.}
\end{dcases}
\end{align}
Hence, from \eqref{eq1: exact DG}--\eqref{eq: compare PMFs},
\begin{align}  \label{eq2: exact DG}
\set{I}_\omega(P_\mu \| P_\lambda) &= \min\{\omega, 1-\omega\}
- \omega \sum_{k=0}^{k_0} P_\mu[k] - (1-\omega) \sum_{k=k_0+1}^{\infty} P_\lambda[k] \\
\label{eq3: exact DG}
&= \min\{\omega, 1-\omega\} - \omega \sum_{k=0}^{k_0} P_\mu[k]
- (1-\omega) \left(1- \sum_{k=0}^{k_0} P_\lambda[k] \right).
\end{align}
To exemplify the utility of the bounds in Theorem~\ref{theorem: DG UBs}, suppose that
$\mu$ and $\lambda$ are close, and we wish to obtain a guarantee on how small
$\set{I}_\omega(P_\mu \| P_\lambda)$ is. For example, let $\lambda = 99$, $\mu = 101$,
and $\omega = \tfrac1{10}$. The upper bounds on $\set{I}_\omega(P_\mu \| P_\lambda)$ in
\eqref{eq: DG-chi^2 UB}--\eqref{eq2: DG-KL UB} are, respectively, equal to $4.6 \cdot 10^{-4}$,
$5.8 \cdot 10^{-4}$ and $2.2 \cdot 10^{-3}$; we therefore get an informative guarantee by
easily calculable bounds. The exact value of $\set{I}_\omega(P_\mu \| P_\lambda)$ is, on
the other hand, hard to compute since $k_0 = 209$ (see \eqref{eq: k_0}), and the calculation
of the right side of \eqref{eq3: exact DG} appears to be sensitive to the selected parameters in this setting.
\end{example}

\section{Local Behavior of $f$-divergences}
\label{section: local behavior}
This section is focused on the local behavior of $f$-divergences;
the starting point relies on \cite[Section~3]{PardoV03} which studies
the asymptotic properties of $f$-divergences. The reader is also
referred to a related study in \cite[Section~4.F]{ISSV16}.

\begin{lemma} \label{lemma: tilted}
Let
\begin{itemize}
\item
$\{P_n\}$ be a sequence of probability measures on a measurable space $(\mathcal{A}, \mathscr{F})$;
\item
the sequence $\{P_n\}$ converge to a probability measure $Q$ in the sense that
\begin{align}
\label{eq: 1st condition}
\lim_{n \to \infty} \text{ess\,sup} \frac{\text{d}P_n}{\text{d}Q} \, (Y) = 1, \quad Y \sim Q
\end{align}
where $P_n \ll Q$ for all sufficiently large $n$;
\item $f, g \in \mathcal{C}$ have continuous second derivatives
at~1 and $g''(1) > 0$.
\end{itemize}
Then
\begin{align} \label{eq: limit of ratio of f-div}
\lim_{n \to \infty} \frac{D_f(P_n \| Q)}{D_g(P_n \| Q)} = \frac{f''(1)}{g''(1)}.
\end{align}
\end{lemma}

\begin{proof}
This follows from \cite[Theorem~3]{PardoV03}, even
without the additional restriction in \cite[Section~3]{PardoV03} which would require
that the second derivatives of $f$ and $g$ are locally Lipschitz at a neighborhood of~1.
More explicitly, in view of the analysis in \cite[p.~1863]{PardoV03}, we get by relaxing
the latter restriction that (cf. \cite[(31)]{PardoV03})
\begin{align}
& \bigl| D_f(P_n \| Q) - \tfrac12 \, f''(1) \, \chi^2(P_n \| Q) \bigr| \nonumber \\
\label{eq1: PardoV}
& \leq \tfrac12 \, \sup_{y \in [1-\varepsilon_n, \, 1+\varepsilon_n]} \bigl| f''(y)-f''(1) \bigr| \; \chi^2(P_n \| Q),
\end{align}
with $\varepsilon_n \downarrow 0$ as we let $n \to \infty$, and also
\begin{align} \label{eq2: PardoV}
\lim_{n \to \infty} \chi^2(P_n \| Q) = 0.
\end{align}
By our assumption, due to the continuity of $f''$ and $g''$ at~1, it follows from \eqref{eq1: PardoV}
and \eqref{eq2: PardoV} that
\begin{align}
\label{eq3: PardoV}
\lim_{n \to \infty} \frac{D_f(P_n \| Q)}{\chi^2(P_n \| Q)} = \tfrac12 \, f''(1), \\[0.2cm]
\label{eq4: PardoV4}
\lim_{n \to \infty} \frac{D_g(P_n \| Q)}{\chi^2(P_n \| Q)} = \tfrac12 \, g''(1),
\end{align}
which yields \eqref{eq: limit of ratio of f-div} (recall that, by assumption, $g''(1)>0$).
\end{proof}

\vspace*{0.1cm}
\begin{remark}
Since $f$ and $g$ in Lemma~\ref{lemma: tilted} are assumed
to have continuous second derivatives at~1, the left and right derivatives
of the weight function $w_f$ in \eqref{eq: weight function} at~1 satisfy,
in view of Remark~\ref{remark 2: w_f},
\begin{align}
w_f'(1^+) = -w_f'(1^-) = f''(1).
\end{align}
Hence, the limit in the right side of \eqref{eq: limit of ratio of f-div}
is equal to $\frac{w_f'(1^+)}{w_g'(1^+)}$ or also to $\frac{w_f'(1^-)}{w_g'(1^-)}$.
\end{remark}

\vspace*{0.1cm}
\begin{lemma} \label{lemma: chi^2 div}
\begin{align}
\label{eq: chi^2 identity}
\chi^2(\lambda P + (1-\lambda)Q \, \| \, Q) = \lambda^2 \, \chi^2(P \| Q), \quad \forall \, \lambda \in [0,1].
\end{align}
\end{lemma}
\begin{proof}
Let $p = \frac{\mathrm{d}P}{\mathrm{d}\mu}$
and $q = \frac{\mathrm{d}Q}{\mathrm{d}\mu}$
be the densities of $P$ and $Q$
with respect to an arbitrary probability measure
$\mu$ such that $P, Q \ll \mu$. Then,
\begin{align}
\chi^2(\lambda P + (1-\lambda)Q \, \| \, Q) &= \int \frac{\bigl( (\lambda p + (1-\lambda)q) - q \bigr)^2}{q} \, \mathrm{d}\mu \\
&= \lambda^2 \int \frac{(p-q)^2}{q}  \, \mathrm{d}\mu \\
&= \lambda^2 \; \chi^2(P \| Q).
\end{align}
\end{proof}

\begin{remark}
The result in Lemma~\ref{lemma: chi^2 div}, for the chi-squared divergence,
is generalized to the identity
\begin{align}
\label{eq: chi^s identity}
\chi^s(\lambda P + (1-\lambda)Q \, \| \, Q) = \lambda^s \, \chi^s(P \| Q), \quad \forall \, \lambda \in [0,1],
\end{align}
for all $s \geq 1$ (see \eqref{eq: chi^s div}). The special case of $s=2$ is required in the continuation of this section.
\end{remark}

\begin{remark}
The result in Lemma~\ref{lemma: chi^2 div} can be generalized as follows:
let $P, Q, R$ be probability measures, and $\lambda \in [0,1]$. Let
$P, Q, R \ll \mu$ for an arbitrary probability measure $\mu$, and
$p := \frac{\mathrm{d}P}{\mathrm{d}\mu}$, $q := \frac{\mathrm{d}Q}{\mathrm{d}\mu}$,
and $r := \frac{\mathrm{d}R}{\mathrm{d}\mu}$ be the corresponding
densities with respect to $\mu$. Calculation shows that
\begin{align} \label{eq: gen. identity chi^2}
\chi^2(\lambda P + (1-\lambda) Q \, \| \, R) - \chi^2(Q \| R)
&= c \lambda + \bigl[ \chi^2(P\|R) - \chi^2(Q\|R) - c\bigr] \lambda^2
\end{align}
with
\begin{align} \label{eq: c}
c & := \int \frac{(p-q)q}{r} \, \mathrm{d}\mu.
\end{align}
If $Q=R$, then $c=0$ in \eqref{eq: c}, and \eqref{eq: gen. identity chi^2}
is specialized to \eqref{eq: chi^2 identity}. However, if $Q \neq R$, then
$c$ may be non-zero. This shows that, for small $\lambda \in [0,1]$, the
left side of \eqref{eq: gen. identity chi^2} scales linearly in $\lambda$
if $c \neq 0$, and it has a quadratic scaling in $\lambda$ if $c=0$ and
$\chi^2(P\|R) \neq \chi^2(Q\|R)$ (e.g., if $Q=R$, as in Lemma~\ref{lemma: chi^2 div}).
The identity in \eqref{eq: gen. identity chi^2} yields
\begin{align}
\frac{\mathrm{d}}{\mathrm{d}\lambda} \, \chi^2(\lambda P + (1-\lambda) Q \, \| \, R) \, \Bigl|_{\lambda=0}
= \lim_{\lambda \downarrow 0} \, \frac{\chi^2(\lambda P + (1-\lambda) Q \, \| \, R) - \chi^2(Q \| R)}{\lambda} = c.
\end{align}
\end{remark}

We next state the main result in this section.
\begin{theorem} \label{theorem: local behavior 2018}
Let
\begin{itemize}
\item $P$ and $Q$ be probability measures defined on a measurable space $(\mathcal{A}, \mathscr{F})$, $Y \sim Q$, and suppose that
\begin{align} \label{eq: bounded RND}
\text{ess\,sup} \frac{\text{d}P}{\text{d}Q} \, (Y) < \infty;
\end{align}
\item $f \in \mathcal{C}$, and $f''$ be continuous at 1.
\end{itemize}
Then,
\begin{align}
\lim_{\lambda \downarrow 0} \frac1{\lambda^2} \; D_f( \lambda P + (1-\lambda)Q \, \| \, Q)
&= \lim_{\lambda \downarrow 0} \frac1{\lambda^2} \; D_f(Q \, \| \, \lambda P + (1-\lambda)Q)
\label{eq1: local behavior} \\[0.1cm]
&= \tfrac12 \, f''(1) \, \chi^2(P \| Q).
\label{eq2: local behavior}
\end{align}
\end{theorem}

\begin{proof}
Let $\{\lambda_n\}_{n \in \ensuremath{\mathbb{N}}}$ be a sequence in $[0,1]$, which tends to zero.
Define the sequence of probability measures
\begin{align} \label{eq: R_n}
R_n := \lambda_n P + (1-\lambda_n) Q, \qquad n \in \ensuremath{\mathbb{N}}.
\end{align}
Note that $P \ll Q$ implies that $R_n \ll Q$ for all $n \in \ensuremath{\mathbb{N}}$. Since
\begin{align}
\frac{\mathrm{d}R_n}{\mathrm{d}Q} = \lambda_n \, \frac{\mathrm{d}P}{\mathrm{d}Q} + (1-\lambda_n),
\end{align}
it follows from \eqref{eq: bounded RND} that
\begin{align}
\lim_{n \to \infty} \text{ess\,sup} \frac{\mathrm{d}R_n}{\mathrm{d}Q} \; (Y) = 1.
\end{align}
Consequently, \eqref{eq3: PardoV} implies that
\begin{align} \label{eq: limit R_n}
\lim_{n \to \infty} \frac{D_f(R_n \| Q)}{\chi^2(R_n \| Q)} = \tfrac12 \, f''(1)
\end{align}
where $\{\lambda_n\}$ in \eqref{eq: R_n} is an arbitrary sequence which tends to zero.
Hence, it follows from \eqref{eq: R_n} and \eqref{eq: limit R_n} that
\begin{align} \label{eq: 101}
\lim_{\lambda \downarrow 0} \frac{D_f(\lambda P + (1-\lambda) Q
\, \| \, Q)}{\chi^2(\lambda P + (1-\lambda) Q \, \| \, Q)}
= \tfrac12 \, f''(1),
\end{align}
and, by combining \eqref{eq: chi^2 identity} and \eqref{eq: 101}, we get
\begin{align} \label{eq1: limit}
\lim_{\lambda \downarrow 0} \frac1{\lambda^2} \; D_f( \lambda P
+ (1-\lambda)Q \, \| \, Q) = \tfrac12 \, f''(1) \, \chi^2(P \| Q).
\end{align}

We next prove the result for the limit in the right side of \eqref{eq1: local behavior}.
Let $f^\ast \colon (0, \infty) \mapsto \ensuremath{\mathbb{R}}$ be the conjugate function of $f$,
which is given in \eqref{eq: conjugate f}. By the assumption that $f$ has a second continuous
derivative, so is $f^\ast$ and it is easy to verify that the second derivatives of $f$
and $f^\ast$ coincide at~1. Hence, from \eqref{eq: Df and Df^ast} and \eqref{eq1: limit},
\begin{align}
\lim_{\lambda \downarrow 0} \frac1{\lambda^2} \; D_f(Q \, \| \, \lambda P + (1-\lambda)Q)
&= \lim_{\lambda \downarrow 0} \frac1{\lambda^2} \; D_{f^\ast}( \lambda P + (1-\lambda)Q \, \| \, Q) \\
&= \tfrac12 \, f''(1) \, \chi^2(P \| Q).
\end{align}
\end{proof}

\begin{remark}
Although an $f$-divergence is in general not symmetric, in the sense that
the equality $D_f(P\|Q) = D_f(Q\|P)$ does not necessarily hold for all pairs
of probability measures $(P,Q)$, the reason for the equality in
\eqref{eq1: local behavior} stems from the fact that the second derivatives
of $f$ and $f^\ast$ coincide at~1 when $f$ is twice differentiable.
\end{remark}

\begin{remark}
Under the conditions in Theorem~\ref{theorem: local behavior 2018}, it follows
from \eqref{eq2: local behavior} that
\begin{align}
\label{first derivative D_f}
& \frac{\mathrm{d}}{\mathrm{d}\lambda} \; D_f(\lambda P + (1-\lambda)Q \, \| \, Q) \, \Bigl|_{\lambda=0}
= \lim_{\lambda \downarrow 0} \frac1\lambda \; D_f(\lambda P + (1-\lambda) Q \, \| \, Q) = 0, \\[0.1cm]
\label{second derivative D_f}
& \lim_{\lambda \downarrow 0} \; \frac{\mathrm{d^2}}{\mathrm{d}\lambda^2} \; D_f(\lambda P + (1-\lambda)Q \, \| \, Q)
= 2 \, \lim_{\lambda \downarrow 0} \frac1{\lambda^2} \; D_f(\lambda P + (1-\lambda) Q \, \| \, Q) = f''(1) \, \chi^2(P\|Q)
\end{align}
where \eqref{second derivative D_f} relies on L'H\^{o}pital's rule.
The convexity of $D_f(P \| Q)$ in $(P,Q)$ also implies that, for all $\lambda \in (0,1]$,
\begin{align}
D_f(\lambda P + (1-\lambda) Q \, \| \, Q) \leq \lambda D_f(P \| Q).
\end{align}
\end{remark}

The following result refers to the local behavior of R\'{e}nyi divergences of an
arbitrary non-negative order.
\begin{corollary}
Under the condition in \eqref{eq: bounded RND}, for every $\alpha \in [0, \infty]$,
\begin{align}
\lim_{\lambda \downarrow 0} \frac1{\lambda^2} \; D_\alpha( \lambda P + (1-\lambda)Q \, \| \, Q)
&= \lim_{\lambda \downarrow 0} \frac1{\lambda^2} \; D_\alpha(Q \, \| \, \lambda P + (1-\lambda)Q)
\label{eq1: local behavior RD} \\[0.1cm]
&= \tfrac12 \, \alpha \, \chi^2(P \| Q) \, \log e.
\label{eq2: local behavior RD}
\end{align}
\end{corollary}

\begin{proof}
Let $\alpha \in (0,1) \cup (1, \infty)$. In view of \eqref{eq: H as fD} and Theorem~\ref{theorem: local behavior 2018},
it follows that the local behavior of the Hellinger divergence of order $\alpha$ satisfies
\begin{align}
\lim_{\lambda \downarrow 0} \frac1{\lambda^2} \; \mathscr{H}_{\alpha}( \lambda P + (1-\lambda)Q \, \| \, Q)
&= \lim_{\lambda \downarrow 0} \frac1{\lambda^2} \; \mathscr{H}_{\alpha}(Q \, \| \, \lambda P + (1-\lambda)Q)
\label{eq1: local behavior Hel} \\[0.1cm]
&= \tfrac12 \, \alpha \, \chi^2(P \| Q).
\label{eq2: local behavior Hel}
\end{align}
The result now follows from \eqref{renyimeetshellinger}, which implies that
\begin{align}
\label{eq1: ratio RD/Hel}
\lim_{\lambda \downarrow 0} \frac{D_{\alpha}( \lambda P + (1-\lambda)Q \, \| \, Q)}{\mathscr{H}_{\alpha}( \lambda P + (1-\lambda)Q \, \| \, Q)}
&= \lim_{\lambda \downarrow 0} \frac{D_{\alpha}(Q \, \| \, \lambda P + (1-\lambda)Q)}{\mathscr{H}_{\alpha}(Q \, \| \, \lambda P + (1-\lambda)Q)} \\
&= \frac1{\alpha-1} \lim_{u \to 0} \frac{\log \bigl(1 + (\alpha-1) u \bigr)}{u} \\
&= \log e. \label{eq2: ratio RD/Hel}
\end{align}
The result in \eqref{eq1: local behavior RD} and \eqref{eq2: local behavior RD}, for $\alpha \in (0,1) \cup (1, \infty)$,
follows by combining the equalities in \eqref{eq1: local behavior Hel}--\eqref{eq2: ratio RD/Hel}.

Finally, the result in \eqref{eq1: local behavior RD} and \eqref{eq2: local behavior RD} for $\alpha \in \{0, 1, \infty\}$
follows from its validity for all $\alpha \in (0,1) \cup (1, \infty)$, and also due to the property where
$D_{\alpha}(\cdot \| \cdot)$ is monotonically increasing in $\alpha$ (see \cite[Theorem~3]{ErvenH14}).
\end{proof}

\appendices

\section{Proof of Theorem~\ref{theorem: some int. representations}}
\label{appendix: proof of identities}

We prove in the following the integral representations of $f$-divergences and
related measures in Theorem~\ref{theorem: some int. representations}.
\begin{enumerate}[1)]
\item Relative entropy: The function $f \in \set{C}$ in \eqref{eq: f for KL}
yields the following weight function in \eqref{eq: weight function}:
\begin{align}
w_f(\beta) = \left( \frac1\beta - \frac1{\beta^2} \right) \left( 1\{\beta \geq 1\} - 1\{0 < \beta < 1\}\right) \, \log e, \quad \beta > 0.
\end{align}
Consequently, setting $c:=\log e$ in \eqref{eq: generalized w_f} yields
\begin{align}
\label{eq: modified w for KL}
\widetilde{w}_{f,c}(\beta) = \frac1\beta \left( 1\{\beta \geq 1\} - 1\{0 < \beta < 1\}\right) \log e,
\end{align}
for $\beta>0$. Equality \eqref{eq: int. rep. KL} follows from the substitution of \eqref{eq: modified w for KL}
into the right side of \eqref{eq2: new int rep Df}.

\item Hellinger divergence: In view of \eqref{eq: Hel-divergence}, for $\alpha \in (0,1) \cup (1, \infty)$,
the weight function $w_{f_\alpha} \colon (0, \infty) \mapsto [0, \infty)$ in \eqref{eq: weight function}
which corresponds to $f_\alpha \colon (0, \infty) \mapsto \Reals$ in \eqref{eq: H as fD} can be
verified to be equal to
\begin{align} \label{eq1: w for Hel}
w_{f_\alpha}(\beta) = \left( \beta^{\alpha-2} - \frac1{\beta^2} \right) \left( 1\{\beta \geq 1\} - 1\{0 < \beta < 1\} \right)
\end{align}
for $\beta>0$. In order to simplify the integral representation of the Hellinger divergence $\mathscr{H}_{\alpha}(P\|Q)$, we
apply Theorem~\ref{theorem: Int. rep.}-\ref{theorem: int. rep. - part 2}). From \eqref{eq1: w for Hel}, setting $c:=1$ in
\eqref{eq: generalized w_f} implies that $\widetilde{w}_{f_\alpha, 1} \colon (0, \infty) \to \Reals$ is given by
\begin{align} \label{eq2: w for Hel}
\widetilde{w}_{f_\alpha, 1}(\beta) = \beta^{\alpha-2} \left( 1\{\beta \geq 1\} - 1\{0 < \beta < 1\} \right)
\end{align}
for $\beta>0$. Hence, substituting \eqref{eq: G function} and \eqref{eq2: w for Hel} into \eqref{eq2: new int rep Df} yields
\begin{align} \label{eq3: int. rep. Hel}
\mathscr{H}_{\alpha}(P\|Q) = \int_1^\infty \beta^{\alpha-2} \, \bigl(1-\mathds{F}_{P\|Q}(\log \beta)\bigr) \, \mathrm{d}\beta
- \int_0^1 \beta^{\alpha-2} \, \mathds{F}_{P\|Q}(\log \beta) \, \mathrm{d}\beta.
\end{align}
For $\alpha > 1$, \eqref{eq3: int. rep. Hel} yields
\begin{align} \label{eq4: int. rep. Hel}
\mathscr{H}_{\alpha}(P\|Q) &= \int_0^\infty \beta^{\alpha-2} \, \bigl(1-\mathds{F}_{P\|Q}(\log \beta)\bigr) \,
\mathrm{d}\beta - \int_0^1 \beta^{\alpha-2} \, \mathrm{d}\beta \\[0.1cm]
&= \int_0^\infty \beta^{\alpha-2} \, \bigl(1-\mathds{F}_{P\|Q}(\log \beta)\bigr) \, \mathrm{d}\beta - \frac1{\alpha-1},
\end{align}
and, for $\alpha \in (0,1)$, \eqref{eq3: int. rep. Hel} yields
\begin{align} \label{eq5: int. rep. Hel}
\mathscr{H}_{\alpha}(P\|Q) &= \int_1^\infty \beta^{\alpha-2} \, \mathrm{d}\beta
- \int_0^\infty \beta^{\alpha-2} \, \mathds{F}_{P\|Q}(\log \beta) \, \mathrm{d}\beta \\[0.1cm]
\label{eq6: int. rep. Hel}
&= \frac1{1-\alpha} - \int_0^\infty \beta^{\alpha-2} \, \mathds{F}_{P\|Q}(\log \beta) \, \mathrm{d}\beta.
\end{align}
This proves \eqref{eq: int. rep. Hel}. We next consider the following special cases:
\begin{itemize}
\item In view of \eqref{eq2: chi^2}, equality \eqref{eq: int. rep. chi^2 div}
readily follows from \eqref{eq: int. rep. Hel} with $\alpha=2$.

\item In view of \eqref{eq3: Sq Hel}, equality \eqref{eq: int. rep. H^2 dist}
readily follows from \eqref{eq: int. rep. Hel} with $\alpha=\tfrac12$.

\item In view of \eqref{eq: B dist}, equality \eqref{eq: int. rep. B dist}
readily follows from \eqref{eq: int. rep. H^2 dist}.
\end{itemize}

\item R\'enyi divergence:
In view of the one-to-one correspondence in \eqref{renyimeetshellinger}
between the R\'enyi divergence and the Hellinger divergence of the same order,
\eqref{eq: int. rep. RenyiD} readily follows from \eqref{eq: int. rep. Hel}.

\item $\chi^s$ divergence with $s \geq 1$: We first consider the case where $s>1$.
From \eqref{eq: chi^s div}, the function $f_s \colon (0, \infty) \mapsto \Reals$
in \eqref{eq: f - chi^s div} is differentiable and $f_s'(1)=0$. Hence, the respective
weight function $w_{f_s} \colon (0, \infty) \mapsto (0, \infty)$ can be verified from
\eqref{eq: weight function} to be given by
\begin{align} \label{eq: weight function chi^s}
w_{f_s}(\beta) = \frac1\beta \left( s-1+\frac1\beta \right) |\beta-1|^{s-1}, \quad \beta>0.
\end{align}
The result in \eqref{eq: int. rep. chi^s}, for $s > 1$, follows readily from
\eqref{eq: chi^s div}, \eqref{eq: G function}, \eqref{eq: new int rep Df} and
\eqref{eq: weight function chi^s}.

We next prove \eqref{eq: int. rep. chi^s} with $s=1$. In view of \eqref{eq: f - chi^s div},
\eqref{eq: f-TV}, \eqref{eq1: TV distance} and the dominated convergence theorem,
\begin{align}
\label{eq1: TV - s=1}
|P-Q| &= \lim_{s \downarrow 1} \, \chi^s(P \| Q) \\
\label{eq2: TV - s=1}
&= \int_1^{\infty} \frac{1-\mathds{F}_{P\|Q}(\log \beta)}{\beta^2} \, \mathrm{d}\beta +
\int_0^1 \frac{\mathds{F}_{P\|Q}(\log \beta)}{\beta^2} \, \mathrm{d}\beta.
\end{align}
This extends \eqref{eq: int. rep. chi^s} for all $s \geq 1$, although $f_1(t)=|t-1|$
for $t>0$ is not differentiable at~1. For $s=1$, in view of \eqref{eq7: identity RIS},
the integral representation in the right side of \eqref{eq2: TV - s=1} can be
simplified to \eqref{eq: int. rep. TV} and \eqref{eq2: int. rep. TV}.

\item DeGroot statistical information:
In view of \eqref{eq:DG f-div}--\eqref{eq: f for DG}, since the function
$\phi_w \colon (0, \infty) \mapsto \Reals$ is not differentiable at the point
$\frac{1-\omega}{\omega} \in (0, \infty)$ for $\omega \in (0,1)$,
Theorem~\ref{theorem: Int. rep.} cannot be applied directly to get an integral
representation of the DeGroot statistical information. To that end,
for $(\omega, \alpha) \in (0,1)^2$, consider the family of convex functions
$f_{\omega, \alpha} \colon (0, \infty) \mapsto \Reals$ given by (see \cite[(55)]{LieseV_IT2006})
\begin{align} \label{LieseV-55}
f_{\omega, \alpha}(t) = \frac1{1-\alpha} \left( \Bigl[ (\omega t )^{\frac1\alpha}
+ (1-\omega)^{\frac1{\alpha}} \Bigr]^\alpha - \Bigl[ \omega^{\frac1\alpha}
+ (1-\omega)^{\frac1{\alpha}} \Bigr]^\alpha \right),
\end{align}
for $t > 0$. These differentiable functions also satisfy
\begin{align}
\label{eq: f --> phi}
\lim_{\alpha \downarrow 0} f_{\omega, \alpha}(t) = \phi_w(t),
\end{align}
which holds due to the identities
\begin{align}
\label{eq1: 2 identities}
& \lim_{\alpha \downarrow 0} \left(a^{\frac1{\alpha}} + b^{\frac1{\alpha}}\right)^\alpha = \max\{a,b\}, \quad \; \; a, b \geq 0; \\
\label{eq2: 2 identities}
& \min\{a,b\} = a+b - \max\{a,b\}, \quad a, b \in \Reals.
\end{align}
The application of Theorem~\ref{theorem: Int. rep.}-\ref{theorem: int. rep. - part 2}) to the set
of functions $f_{\omega, \alpha} \in \set{C}$ with
\begin{align}
c:= \frac{(1-\omega)^{\frac1\alpha}}{\alpha-1} \left[ \omega^{\frac1\alpha}
+ (1-\omega)^{\frac1{\alpha}} \right]^{\alpha-1}
\end{align}
yields
\begin{align}
\label{eq: w - Arimoto info.}
\widetilde{w}_{f_{\omega, \alpha}, \, c}(\beta) = \frac{1-\omega}{1-\alpha} \,
\frac1{\beta^2} \left[ 1 + \left( \frac{\omega \beta}{1-\omega} \right)^{\frac1\alpha} \right]^{\alpha-1} \;
\Bigl[ 1\{0<\beta<1\} - 1\{\beta \geq 1\} \Bigr],
\end{align}
for $\beta>0$, and
\begin{align}
\label{eq: Arimoto info.}
D_{f_{\omega, \alpha}}(P\|Q) = \int_0^{\infty} \widetilde{w}_{f_{\omega, \alpha}, \, c}(\beta) \, G_{P\|Q}(\beta) \, \mathrm{d}\beta
\end{align}
with $G_{P\|Q}(\cdot)$ as defined in \eqref{eq: G function}, and $(\omega, \alpha) \in (0,1)^2$.
From \eqref{eq1: 2 identities} and \eqref{eq: w - Arimoto info.}, it follows that
\begin{align}
\label{eq: limit of w for DeGroot info.}
\lim_{\alpha \downarrow 0} \, \widetilde{w}_{f_{\omega, \alpha}, \, c}(\beta)
= \frac{1-\omega}{\beta^2} \; \Bigl[ 1\{0<\beta<1\} - 1\{\beta \geq 1\} \Bigr] \;
\Bigl[ \tfrac12 \, 1\bigl\{\beta = \tfrac{1-\omega}{\omega}\bigr\} + 1\bigl\{0<\beta<\tfrac{1-\omega}{\omega}\bigr\} \Bigr],
\end{align}
for $\beta>0$. In view of \eqref{eq:DG f-div}, \eqref{eq: f for DG}, \eqref{eq: G function}, \eqref{eq: f --> phi},
\eqref{eq: Arimoto info.} and \eqref{eq: limit of w for DeGroot info.}, and the monotone convergence theorem,
\begin{align}
\mathcal{I}_{\omega}(P\|Q) &= D_{\phi_\omega}(P\|Q) \\
&=\lim_{\alpha \downarrow 0} D_{f_{\omega, \alpha}}(P\|Q)\\
\label{eq: to be simplified}
&= (1-\omega) \int_0^{\min\{1, \frac{1-\omega}{\omega}\}} \; \frac{\mathds{F}_{P\|Q}(\log \beta)}{\beta^2} \, \mathrm{d}\beta
- (1-\omega) \int_1^{\max\{1, \frac{1-\omega}{\omega}\}} \; \frac{1-\mathds{F}_{P\|Q}(\log \beta)}{\beta^2} \, \mathrm{d}\beta,
\end{align}
for $\omega \in (0,1)$. We next simplify \eqref{eq: to be simplified} as follows:
\begin{enumerate}[a)]
\item if $\omega \in \bigl(1, \tfrac{1-\omega}{\omega} \bigr)$, then $\tfrac{1-\omega}{\omega} < 1$ and \eqref{eq: to be simplified} yields
\begin{align}
\mathcal{I}_{\omega}(P\|Q) &= (1-\omega) \int_0^{\frac{1-\omega}{\omega}} \; \frac{\mathds{F}_{P\|Q}(\log \beta)}{\beta^2} \, \mathrm{d}\beta;
\end{align}
\item if $\omega \in \bigl(0, \tfrac12 \bigr]$, then $\tfrac{1-\omega}{\omega} \geq 1$ and \eqref{eq: to be simplified} yields
\begin{align}
\mathcal{I}_{\omega}(P\|Q) &= (1-\omega) \int_0^1 \frac{\mathds{F}_{P\|Q}(\log \beta)}{\beta^2} \, \mathrm{d}\beta -
(1-\omega) \int_1^{\frac{1-\omega}{\omega}} \frac{1-\mathds{F}_{P\|Q}(\log \beta)}{\beta^2} \, \mathrm{d}\beta \\[0.2cm]
\label{eq: DeGroot + Lemma}
&= (1-\omega)  \int_1^\infty \frac{1-\mathds{F}_{P\|Q}(\log \beta)}{\beta^2} \, \mathrm{d}\beta -
(1-\omega) \int_1^{\frac{1-\omega}{\omega}} \frac{1-\mathds{F}_{P\|Q}(\log \beta)}{\beta^2} \, \mathrm{d}\beta \\[0.2cm]
&= (1-\omega) \int_{\frac{1-\omega}{\omega}}^\infty \frac{1-\mathds{F}_{P\|Q}(\log \beta)}{\beta^2} \, \mathrm{d}\beta,
\end{align}
where \eqref{eq: DeGroot + Lemma} follows from \eqref{eq7: identity RIS} (or its equivalent from in \eqref{eq: int. identity with RIS}).
\end{enumerate}
This completes the proof of \eqref{eq: int. rep. DeGroot Info}. Note that, due to \eqref{eq7: identity RIS},
the integral representation of $\mathcal{I}_{\omega}(P\|Q)$ in \eqref{eq: int. rep. DeGroot Info} is indeed
continuous at $\omega=\tfrac12$.

\item Triangular discrimination: In view of \eqref{eq:delta}--\eqref{eq:tridiv},
the corresponding function $\widetilde{w}_{f,1} \colon (0, \infty) \mapsto \Reals$ in
\eqref{eq: generalized w_f} (i.e., with $c:=1$) can be verified to be given by
\begin{align} \label{eq: w for delta div}
\widetilde{w}_{f,1}(\beta) = \frac{4}{(\beta+1)^2} \left(1\{\beta \geq 1\} - 1\{0 < \beta < 1\} \right)
\end{align}
for $\beta>0$. Substituting \eqref{eq: G function} and \eqref{eq: w for delta div} into
\eqref{eq2: new int rep Df} proves \eqref{eq: int. rep. TD} as follows:
\begin{align}
\Delta(P\|Q) &= 4 \left( \int_1^\infty \frac{1-\mathds{F}_{P\|Q}(\log \beta)}{(\beta+1)^2} \, \mathrm{d}\beta
- \int_0^1 \frac{\mathds{F}_{P\|Q}(\log \beta)}{(\beta+1)^2}  \, \mathrm{d}\beta \right) \\[0.3cm]
&= 4 \left( \int_0^\infty \frac{1-\mathds{F}_{P\|Q}(\log \beta)}{(\beta+1)^2} \, \mathrm{d}\beta
- \int_0^1 \frac1{(\beta+1)^2}  \, \mathrm{d}\beta \right) \\[0.3cm]
&= 4 \int_0^\infty \frac{1-\mathds{F}_{P\|Q}(\log \beta)}{(\beta+1)^2} \, \mathrm{d}\beta - 2.
\end{align}

\item Lin's measure and the Jensen-Shannon divergence: Let $\theta \in (0,1)$
(if $\theta \in \{0, 1\}$, then \eqref{eq: Lin91}--\eqref{eq2: Lin91}
imply that $L_\theta(P\|Q)=0$). In view of \eqref{eq: Lin div as Df},
the application of Theorem~\ref{theorem: Int. rep.}-\ref{theorem: int. rep. - part 1})
with the function $f_{\theta} \colon (0, \infty) \mapsto \Reals$ in
\eqref{eq: f of Lin div} yields the weight function
$w_{f_\theta} \colon (0, \infty) \mapsto [0, \infty)$ defined as
\begin{align} \label{eq: w_f - Lin}
w_{f_\theta}(\beta) &= \frac{(1-\theta) \, \log(\theta \beta + 1 - \theta)}{\beta^2}
\; \Bigl(1\{\beta \geq 1\} - 1\{0 < \beta < 1\}\Bigr).
\end{align}
Consequently, we get
\begin{align}
\label{eq11: Lin int. rep.}
L_{\theta}(P\|Q) = & \, (1-\theta) \left( \int_1^\infty \tfrac{\log(\theta \beta + 1 - \theta)}{\beta^2}
\left(1 - \mathds{F}_{P \| Q}(\log \beta)\right) \, \mathrm{d}\beta
- \int_0^1 \tfrac{\log(\theta \beta + 1 - \theta)}{\beta^2} \;
\mathds{F}_{P \| Q}(\log \beta) \, \mathrm{d}\beta \right) \\[0.3cm]
\label{eq12: Lin int. rep.}
= & \, (1-\theta) \left( \int_1^\infty \frac{\log(\theta \beta + 1 - \theta)}{\beta^2} \, \mathrm{d}\beta
- \int_0^\infty \frac{\log(\theta \beta + 1 - \theta)}{\beta^2} \;
\mathds{F}_{P \| Q}(\log \beta) \, \mathrm{d}\beta \right) \\[0.3cm]
\label{eq13: Lin int. rep.}
= & \, \theta \log \frac1\theta - (1-\theta)  \int_0^\infty \frac{\log(\theta \beta + 1 - \theta)}{\beta^2} \;
\mathds{F}_{P \| Q}(\log \beta) \, \mathrm{d}\beta \\[0.3cm]
\label{eq14: Lin int. rep.}
= & \, h(\theta) - (1-\theta)  \int_0^\infty
\frac1{\beta^2} \, \log\left(\frac{\theta \beta}{1 - \theta} + 1\right) \;
\mathds{F}_{P \| Q}(\log \beta) \, \mathrm{d}\beta
\end{align}
where \eqref{eq11: Lin int. rep.} follows from \eqref{eq: G function}, \eqref{eq: new int rep Df}
and \eqref{eq: w_f - Lin}; for $\theta \in (0,1)$, equality \eqref{eq13: Lin int. rep.} holds since
\begin{align}
\label{eq15: Lin int. rep.}
\int_1^\infty \frac{\log(\theta \beta + 1 - \theta)}{\beta^2} \, \mathrm{d}\beta
= \frac{\theta}{1-\theta} \, \log \frac1{\theta};
\end{align}
finally, \eqref{eq14: Lin int. rep.} follows from \eqref{eq: int. identity with RIS} where
$h \colon [0,1] \mapsto [0,\log 2]$ denotes the binary entropy function. This proves
\eqref{eq: int. rep. Lin's div}. In view of \eqref{eq:js1}, the identity in \eqref{eq: int. rep. JS div}
for the Jensen-Shannon divergence follows from \eqref{eq: int. rep. Lin's div} with $\theta = \tfrac12$.

\item Jeffrey's divergence: In view of \eqref{eq2: jeffreys}--\eqref{f - jeffreys}, the corresponding
weight function $w_f \colon (0, \infty) \mapsto [0, \infty)$ in \eqref{eq: weight function} can be
verified to be given by
\begin{align} \label{eq1: w for jeffreys}
w_f(\beta) = \left( \frac{\log e}{\beta} + \frac1{\beta^2} \, \log \frac{\beta}{e} \right)
\left( 1\{\beta \geq 1\} - 1\{0<\beta<1\} \right).
\end{align}
Hence, setting $c := \log e$ in \eqref{eq: generalized w_f} implies that
\begin{align} \label{eq2: w for jeffreys}
\widetilde{w}_{f,c}(\beta)
= \left( \frac{\log e}{\beta} + \frac{\log \beta}{\beta^2} \right) \left( 1\{\beta \geq 1\} - 1\{0<\beta<1\} \right)
\end{align}
for $\beta > 0$. Substituting \eqref{eq: G function} and \eqref{eq2: w for jeffreys} into
\eqref{eq2: new int rep Df} yields \eqref{eq: int. rep. Jefreey's div}.

\item $E_\gamma$ divergence: Let $\gamma \geq 1$, and let
$\omega \in \bigl(0, \tfrac12 \bigr]$ satisfy $\frac{1-\omega}{\omega} = \gamma$; hence,
$\omega = \frac1{1+\gamma}$. From \eqref{eq: DG-EG}, we get
\begin{align} \label{eq1: EG-DG}
E_{\gamma}(P\|Q) = (1+\gamma) \, \mathcal{I}_{\frac1{1+\gamma}}(P\|Q).
\end{align}
The second line in the right side of \eqref{eq: int. rep. DeGroot Info} yields
\begin{align} \label{eq2: EG-DG}
\mathcal{I}_{\frac1{1+\gamma}}(P\|Q) = \frac{\gamma}{1+\gamma}
\int_\gamma^\infty \frac{1-\mathds{F}_{P\|Q}(\log \beta)}{\beta^2} \, \mathrm{d}\beta.
\end{align}
Finally, substituting \eqref{eq2: EG-DG} into the right side of \eqref{eq1: EG-DG}
yields \eqref{eq: int. rep. E_gamma}.
\end{enumerate}

\begin{remark}
In view of \eqref{eq7: identity RIS}, the integral representation
for the $\chi^s$ divergence in \eqref{eq: int. rep. chi^s} specializes to
\eqref{eq: int. rep. chi^2 div} and \eqref{eq: int. rep. TV}--\eqref{eq2: int. rep. TV}
by letting $s=2$ and $s=1$, respectively.
\end{remark}

\begin{remark}
In view of \eqref{eq:EG-TV}, the first identity for the total variation distance
in \eqref{eq: int. rep. TV} follows readily from \eqref{eq: int. rep. E_gamma}
with $\gamma=1$. The second identity in \eqref{eq2: int. rep. TV} follows from
\eqref{eq: int. identity with RIS} and \eqref{eq: int. rep. TV}, and since
$\int_1^\infty \frac{\mathrm{d}\beta}{\beta^2} = 1$.
\end{remark}

\section*{Acknowledgment}
The author is grateful to Sergio Verd\'{u} and the two anonymous reviewers, whose suggestions improved the
presentation in this paper.


\end{document}